\documentclass[11pt]{llncs}
\usepackage{amsmath,amssymb}
\usepackage{graphicx}
\usepackage{fullpage}
\usepackage[top=0.8in, bottom=0.9in, left=0.9in, right=0.9in]{geometry}

\usepackage{lineno}

\def\calC{\mathcal{C}}
\def\calH{\mathcal{H}}

\begin{document}

\title{Algorithms for the Line-Constrained Disk Coverage and Related Problems\thanks{This research was supported in part by NSF under Grant CCF-2005323. A preliminary version of this paper will appear in {\em Proceedings of the 17th Algorithms and Data Structures Symposium (WADS 2021)}.}}
\author{Logan Pedersen 
\and
Haitao Wang 
}

 \institute{
 Department of Computer Science\\
  Utah State University, Logan, Utah 84322, USA\\
  \email{logan.pedersen@aggiemail.usu.edu, haitao.wang@usu.edu}
}

\maketitle

\pagestyle{plain}
\pagenumbering{arabic}
\setcounter{page}{1}

\vspace{-0.1in}
\begin{abstract}
Given a set $P$ of $n$ points and a set $S$ of $m$ weighted disks in the
plane, the disk coverage problem asks for a subset of disks of minimum
total weight that cover all points of $P$. The problem is
NP-hard. In this paper, we consider a line-constrained version in
which all disks are centered on a line $L$ (while points of $P$ can be anywhere in the plane). We present an
$O((m+n)\log(m+n)+\kappa\log m)$ time algorithm for the problem, where
$\kappa$ is
the number of pairs of disks that intersect. Alternatively, we can also solve the problem in $O(nm\log(m+n))$ time. For the unit-disk case where
all disks have the same radius, the running time can be reduced to $O((n+m)\log(m+n))$.
In addition, we solve in $O((m+n)\log(m+n))$ time the
$L_{\infty}$ and $L_1$ cases of the problem, in which the disks are
squares and diamonds, respectively.
As a by-product, the 1D version of the problem where all points of $P$ are on $L$
and the disks are line segments on $L$ is also solved in $O((m+n)\log(m+n))$ time.
We also show that the problem has an $\Omega((m+n)\log (m+n))$ time lower bound even for the 1D case.

We further demonstrate that our techniques can also be used to solve
other geometric coverage problems. For example, given in the plane a
set $P$ of $n$ points and a set $S$ of $n$ weighted half-planes, we
solve in $O(n^4\log n)$ time the problem of finding a subset of
half-planes to cover $P$ so that their total weight is minimized. This improves the previous
best algorithm of $O(n^5)$ time
by almost a linear factor. If all half-planes are lower ones, then our algorithm
runs in $O(n^2\log n)$ time, which improves the previous best
algorithm of $O(n^4)$ time
by almost a quadratic factor.
\end{abstract}

\section{Introduction}
\label{sec:intro}

Given a set $P$ of $n$ points and a set $S$ of $m$ disks in the
plane such that each disk has a weight,
the {\em disk coverage} problem asks for a subset of disks of minimum
total weight that cover all points of $P$.
We assume that the union of all disks covers all points of $P$.
It is known that the problem is NP-hard~\cite{ref:FederOp88} and many
approximation algorithms have been proposed,
e.g.,~\cite{ref:LiA15,ref:MustafaPt09}.

In this paper, we consider a line-constrained version of the
problem in which all disks (possibly with different radii) have their centers on a line $L$, say, the
$x$-axis. To the best of our knowledge, this line-constrained problem
was not particularly studied before. We present an
$O((m+n)\log(m+n)+\kappa\log m)$ time algorithm, where $\kappa$ is
the number of pairs of disks that intersect (and thus $\kappa\leq m(m-1)/2$; e.g., if the disks are disjoint, then $\kappa=0$ and the algorithm runs in $O((m+n)\log(m+n))$ time). Alternatively, we can also solve the problem in $O(nm\log(m+n))$ time.
For the {\em unit-disk case} where
all disks have the same radius, the running time can be reduced to $O((n+m)\log(m+n))$.
In addition, we solve in $O((m+n)\log(m+n))$ time the
$L_{\infty}$ and $L_1$ cases of the problem, in which the disks are
squares and diamonds, respectively.
As a by-product, we present an  $O((m+n)\log (m+n))$ time algorithm
for the 1D version of the problem where all points of $P$ are on $L$
and the disks are line segments of $L$.
In addition, we show that the problem has an $\Omega((m+n)\log (m+n))$ time lower bound in the algebraic decision tree model even for the 1D case. This implies that our algorithms for the 1D, $L_{\infty}$, $L_1$, and unit-disk cases are all optimal.

Our algorithms potentially have applications, e.g., in facility locations. For example, suppose we want to build some facilities along a railway which is represented by $L$ (although an entire railway may not be a straight line, it may be considered straight in a local region) to provide service for some customers that are represented by the points of $P$. The center of a disk represents a candidate location for building a facility that can serve the customers covered by the disk and the cost for building the facility is the weight of the disk. The problem is to determine the best locations to build facilities so that all customers can be served and the total cost is minimized. This is exactly an instance of our problem.

Although the problems are line-constrained, our techniques can actually be used to solve other geometric coverage
problems. If all disks of $S$ have the same radius and the set of disk centers are separated from $P$ by a line $\ell$, the problem is called {\em line-separable unit-disk coverage}. The unweighted case of the problem where the weights of all disks are $1$ has been studied in the literature~\cite{ref:AmbuhlCo06,ref:ClaudeAn10,ref:ClaudePr09}. In particular, the fastest algorithm was given by Claude et al.~\cite{ref:ClaudeAn10} and the runtime is $O(n\log n + nm)$. The algorithm, however, does not work for the weighted case.
Our algorithm for the line-constrained $L_2$ case can be used to solve the weighted case in $O(nm\log (m+n))$ time or in $O((m+n)\log (m+n)+\kappa\log m)$ time, where $\kappa$ is the number of pairs of disks that intersect on the side of $\ell$ that contains $P$.
More interestingly, we can use the algorithm to solve the following
{\em half-plane coverage problem}. Given in the plane a set $P$ of $n$
points and a set $S$ of $m$ weighted half-planes, find a subset of the
half-planes to cover all points of $P$ so that their total weight is
minimized. For the {\em lower-only case} where all half-planes are
lower ones, Chan and Grant~\cite{ref:ChanEx14} gave an $O(mn^2(m+n))$
time algorithm. In light of the observation that a half-plane is a
special disk of infinite radius, our line-separable unit-disk
coverage algorithm can be applied to solve the problem
in $O(nm\log (m+n))$ time or in $O(n\log n+m^2\log m)$ time. This
improves the result of~\cite{ref:ChanEx14} by almost a quadratic
factor (note that the techniques of~\cite{ref:ChanEx14} are applicable to more general problem settings such as downward shadows of $x$-monotone curves). For the general case where both upper and lower half-planes
are present, Har-Peled and Lee~\cite{ref:Har-PeledWe12} proposed an
algorithm of $O(n^5)$ time when $m=n$. By using our lower-only case
algorithm, we solve the problem in $O(n^3m\log (m+n))$ time or in
$O(n^3\log n+n^2m^2\log m)$ time. Hence, our result improves the one
in~\cite{ref:Har-PeledWe12} by almost a linear factor.
We believe that our techniques may have other applications that remain
to be discovered.


\subsection{Related work}

Our problem is a new type of set cover problem.
The general set cover problem, which is fundamental and has
been studied extensively, is hard to solve, even
approximately~\cite{ref:FeigeA98,ref:Hockbaum87,ref:LundOn94}. Many
set cover problems in geometric settings, often called geometric
coverage problems, are also NP-hard,
e.g.,~\cite{ref:ChanEx14,ref:Har-PeledWe12}. As mentioned above, if
the line-constrained condition is dropped, then the disk coverage
problem becomes NP-hard, even if all disks are unit disks with the
same weight~\cite{ref:FederOp88}. Polynomial time approximation
schemes (PTAS) exist for the unweighted problem~\cite{ref:MustafaPt09}
as well as the weighted unit-disk case~\cite{ref:LiA15}.


Alt et al.~\cite{ref:AltMi06} studied a problem closely related to
ours, with the same input, consisting of $P$, $S$, and $L$, and the
objective is also to find a subset of disks of minimum total weight
that cover all points of $P$. But the difference is that $S$ is comprised
of all possible disks centered at $L$ and the weight of each disk is
defined as $r^{\alpha}$ with $r$ being the radius of the disk and
$\alpha$ being a given constant at least $1$. Alt et
al.~\cite{ref:AltMi06} gave an $O(n^4\log n)$ time algorithm for any
$L_{p}$ metric and any $\alpha\geq 1$, an $O(n^2\log n)$ time
algorithm for any $L_{p}$ metric and $\alpha=1$, and an $O(n^3\log n)$
time algorithm for the $L_{\infty}$ metric and any $\alpha\geq 1$.
Recently, Pedersen and Wang~\cite{ref:PedersenOn18} improved all these
results by providing an $O(n^2)$ time algorithm for any $L_{p}$ metric
and any $\alpha\geq 1$. A 1D variation of the problem was studied
in the literature where points of $P$ are all on $L$ and another set
$Q$ of $m$ points is given on $L$ as the only candidate centers
for disks. Bil\`o et al.~\cite{ref:BiloGe05} first showed that the
problem is solvable in polynomial time. Lev-Tov and
Peleg~\cite{ref:Lev-TovPo05} gave an algorithm of $O((n+m)^3)$
time for any $\alpha\geq 1$. Biniaz et al.~\cite{ref:BiniazFa18}
recently proposed an $O((n+m)^2)$ time algorithm for the case $\alpha=1$.
Pedersen and Wang~\cite{ref:PedersenOn18} solved the problem in	
$O(n(n+m)+m\log m)$ time for any $\alpha\geq 1$.

Other line-constrained problems have also been studied in the literature, e.g.,~\cite{ref:KarmakarSo13,ref:WangLi16}.


\subsection{Our approach}

We first solve the 1D version of the line-constrained problem by a simple dynamic programming algorithm.
Then, for the general ``1.5D'' problem (i.e., points of $P$ are in the plane), a key observation is that if the points
of $P$ are sorted by their $x$-coordinates, then
the sorted list can be partitioned into sublists such that there
exists an optimal solution in which each disk covers a sublist.
Based on the observation, we reduce the 1.5D problem to an instance
of the 1D problem with a set $P'$ of $n$ points and a set $S'$ of
segments. Two challenges arise in our approach.

The first challenge is to give a small bound on the size of
$S'$. A straightforward method shows that $|S'|\leq n\cdot m$. In the
unit-disk case and the $L_1$ case, we prove that $|S'|$ can be
reduced to $m$ by similar methods. In the $L_{\infty}$ case,
with a different technique, we show that $|S'|$ can be bounded by $2(n+m)$. The most challenging
case is the $L_2$ case. By a number of observations, we prove that
$|S'|\leq 2(n+m)+\kappa$.

The second challenge of our approach is to compute the set $S'$ (the
set $P'$, which actually consists of all projections of the points of
$P$ onto $L$, can be easily obtained in $O(n)$ time).
Our algorithms for computing $S'$ for all cases
use the sweeping technique. The algorithms for the unit-disk case and the
$L_1$ case are relatively easy, while those for the $L_{\infty}$ and
$L_2$ cases require much more effort. Although the two
algorithms for $L_{\infty}$ and $L_2$ are similar in spirit, the
intersections of the disks in the $L_2$ case bring more difficulties and make
the algorithm more involved and less efficient. In summary, computing
$S'$ can be done in $O((n+m)\log(n+m))$ time for all cases except
the $L_2$ case which takes $O((n+m)\log(n+m)+\kappa\log m)$ time.

\paragraph{\bf Outline.} The rest of the paper is organized as follows. We define some notation in Section~\ref{sec:pre} and we present our algorithm for the 1D problem in Section~\ref{sec:1d}.
The unit-disk case and the $L_1$ case are discussed in Section~\ref{sec:unit} and
Section~\ref{sec:l1}, respectively. The algorithms for the $L_{\infty}$ and $L_2$ cases are given in
Section~\ref{sec:general}. Using the algorithm for the $L_2$ case, we solve the line-separable disk coverage problem and the half-plane coverage problem in Section~\ref{sec:line-separable}.
Section~\ref{sec:conclude} concludes the paper with a lower bound proof.

\section{Preliminaries}
\label{sec:pre}

We assume that $L$ is the $x$-axis. We also assume that all points of $P$ are above or on $L$ since otherwise if a point $p_i$ is below $L$, then we could obtain the same optimal solution by replacing $p_i$ with its symmetric point with respect to $L$. For ease of exposition, we make a general position assumption that no two points of $P$ have the same $x$-coordinate and no point of $P$ lies on the boundary of a disk of $S$.

For any point $p$ in the plane, we use $x(p)$ and $y(p)$ to refer to its $x$-coordinate and $y$-coordinate, respectively.

We sort all points of $P$ by their $x$-coordinates, and let $p_1,p_2,\ldots,p_n$ be the sorted list from left to right on $L$. For any $1\leq i\leq j\leq n$, let $P[i,j]$ denote the subset $\{p_i,p_{i+1},\ldots,p_j\}$.
Sometimes we use indices to refer to points of $P$. For example, point $i$ refers to $p_i$.

We sort all disks of $S$ by the $x$-coordinates of their centers from left to right, and let $s_1,s_2,\ldots,s_m$ be the sorted list. For each disk $s_i$, we use $c_i$ to denote its center and use $w_i$ to denote its weight. We assume that each $w_i$ is positive (otherwise one could always include $s_i$ in the solution).
For each disk $s_i$, let $l_i$ and $r_i$ refer to its leftmost and rightmost points, respectively.

We often talk about the relative positions of two geometric objects $O_1$ and $O_2$ (e.g., two points, or a point and a line).  We say that $O_1$ is to the {\em left} of $O_2$ if $x(p)\leq x(p')$ holds for any point $p\in O_1$ and any point $p'\in O_2$, and {\em strictly left} means $x(p)< x(p')$. Similarly, we can define {\em right, above, below}, etc.

For convenience, we use $p_0$ (resp., $p_{n+1}$) to denote a point on $L$ strictly to the left (resp. right) of all points of $P$ and all disks of $S$.

We use the term {\em optimal solution subset} to refer to a subset of $S$ used in an optimal solution, and the {\em optimal objective value} refers to the total sum of the weights of the disks in an optimal solution subset.

\section{The 1D problem}
\label{sec:1d}

In the 1D problem, each disk $s_i\in S$ is a line segment on $L$, and thus $l_i$ and $r_i$ are the left and right endpoints of $s_i$, respectively.
We present a simple dynamic programming algorithm for the problem. We first introduce some notation.

For each segment $s_j\in S$, let $f(j)$ refer to the index of the rightmost point of $P\cup\{p_0\}$ strictly to the left of $l_j$, i.e., $f(j)=\arg\max_{0\leq i\leq n}x(p_i)<x(l_j)$. Due to the definition of $p_0$, $f(j)$ is well defined.
The indices $f(j)$ for all $j=1,2,\ldots,m$ can be obtained in $O(n+m)$ time after we sort all points of $P$ along with the left endpoints of all segments of $S$.

For each $i\in [1,n]$, let $W(i)$ denote the minimum total weight of a subset of disks of $S$ covering all points of $P[1,i]$. Our goal is to compute $W(n)$. For convenience, we set $W(0)=0$.
For each segment $s_j\in S$, we define its {\em cost} as $cost(j)=w_j+W(f(j))$.
One can verify that $W(i)$ is equal to the minimum $cost(j)$ among all segments $s_j\in S$ that cover $p_i$. This is the recursive relation of our dynamic programming algorithm.

We sweep a point $q$ on $L$ from left to right. Initially, $q$ is at $p_0$. During the sweeping, we maintain a subset $S(q)$ of segments that cover $q$, and the cost of each segment of $S(q)$ is already known. Also, the values $W(i)$ for all points $p_i\in P$ to the left of $q$ have been computed.
An event happens when $q$ encounters an endpoint of a segment of $S$ or a point of $P$. To guide the sweeping, we sort all endpoints of the segments of $S$ along with the points of $P$.

If $q$ encounters a point $p_i\in P$, then we find the segment of $S(q)$ with the minimum cost and assign the cost to $W(i)$. If $q$ encounters the left endpoint of a segment $s_j$, we set $cost(j)=w_j+W(f(j))$ and then insert $s_j$ into $S(q)$. If $q$ encounters the right endpoint of a segment, we remove the segment from $S(q)$. If we maintain the segments of $S(q)$ by a balanced binary search tree with their costs as keys, then processing each event takes $O(\log m)$ time as $|S(q)|\leq m$.

Therefore, the sweeping takes $O((n+m)\log m)$ time, after sorting the points of $P$ and all segment endpoints in $O((n+m)\log(n+m))$ time. After the sweeping, $W(n)$ is the optimal objective value, and an optimal solution subset of $S$ can be obtained by the standard back-tracking technique, and we omit the details.

\begin{theorem}\label{theo:1d}
The 1D disk coverage problem is solvable in $O((n+m)\log(n+m))$ time.
\end{theorem}


\section{The unit-disk case}
\label{sec:unit}

In this case, all disks of $S$ have the same radius. We will reduce
the problem to an instance of the 1D problem and then apply
Theorem~\ref{theo:1d}. To this end, we will need to present several observations.

For each disk $s_i$, among all points of $P\cup\{p_0,p_{n+1}\}$ to the right of its center $c_i$,
define $a_r(i)$ as the index of the leftmost point outside $s_i$ (e.g., see Fig.~\ref{fig:toprightpoint}).
Similarly, among all points of $P\cup\{p_0,p_{n+1}\}$ to the left of $c_i$, define
$a_l(i)$ as the index of the rightmost point outside $s_i$. Note that $a_r(i)$ and $a_l(i)$ are well defined due to $p_0$ and $p_{n+1}$.
If $a_l(i)+1<a_r(i)$, then we say that $s_i$ is a {\em useful} disk.

\begin{figure}[t]
\begin{minipage}[t]{\textwidth}
\begin{center}
\includegraphics[height=1.0in]{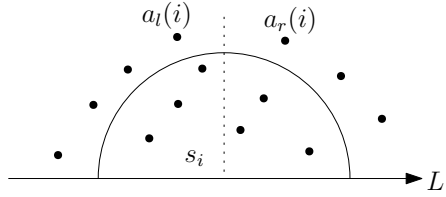}
\caption{\footnotesize Illustrating the two points $a_r(i)$ and $a_l(i)$. The black points are points of $P$. The vertical line is the one through the center of $s_i$. Only the upper half disk of $s_i$ is shown.}
\label{fig:toprightpoint}
\end{center}
\end{minipage}
\vspace{-0.15in}
\end{figure}

Let $P(s_i)$ denote the subset of points of $P$ that are covered by
$s_i$. We further partition $P(s_i)$ into three
subsets as follows. Let $P_l(s_i)$ consist of the points of
$P(s_i)$ strictly to the left of point $a_l(i)$. Let $P_r(s_i)$ consist of the points of
$P(s_i)$ strictly to the right of point $a_r(i)$. Let
$P_m(s_i)=P(s_i)\setminus \{ P_l(s_i)\cup P_r(s_i)\}$.
Observe that $P_m(s_i)\neq \emptyset$ if and only if $s_i$ is a
useful disk, and if $s_i$ is a useful disk, then $P_m(s_i)=P[a_l(i)+1,a_r(i)-1]$.

The following lemma is due to the fact that all disks of $S$ have the same radius and are centered at $L$.

\begin{lemma}\label{lem:unitcover}
Consider a disk $s_i$. If another disk $s_j$ covers the point $a_r(i)$,
then $s_j$ covers all points of $P_r(s_i)$; similarly,
if another disk $s_j$ covers the point $a_l(i)$, then $s_j$ covers all points of $P_l(s_i)$.
\end{lemma}
\begin{proof}
We only prove the case for $a_r(s_i)$, since the other case is
similar. Let $k=a_r(i)$. Assume that a disk $s_j$ covers the point
$p_k$. Our goal is to prove that $s_j$ covers all points of
$P_r(s_i)$. This is obviously true if $P_r(s_i)=\emptyset$. In the
following, we assume that $P_r(s_i)\neq \emptyset$. This implies that
$x(p_k)<x(r_i)$, where $r_i$ is the rightmost point of $s_i$.
Also, by definition, we have $x(c_i)\leq x(p_k)$, where $c_i$ is the
center of $s_i$.

\begin{figure}[h]
\begin{minipage}[t]{\textwidth}
\begin{center}
\includegraphics[height=0.9in]{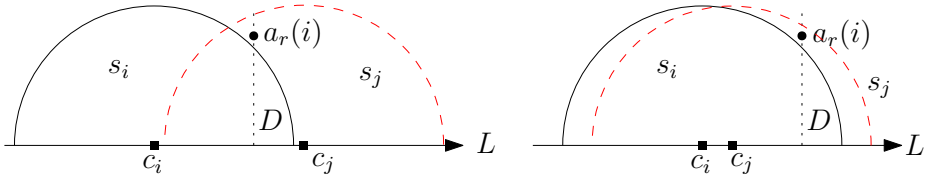}
\caption{\footnotesize Illustrating the proof of Lemma~\ref{lem:unitcover}. The red dashed half-circle is $s_j$ and the black solid half-circle is $s_i$. The two squares on $L$ are the centers of the two disks. Left: $c_j$ is to the right of point $a_r(i)$. Right: $c_j$ is to the left of point $a_r(i)$. In both cases, disk $s_j$ contains the region $D$. }
\label{fig:contain}
\end{center}
\end{minipage}
\vspace{-0.15in}
\end{figure}

Let $D$ be the region of $s_i$ to the right of the vertical line
through $p_k$. By definition, $P_r(s_i)=D\cap P$.
Since $s_i$ and $s_j$ have the same radius and $s_j$
covers $p_k$ while $s_i$ does not, one can verify that $D$ must be
contained in the disk $s_j$, regardless of whether $c_j$ is to the
left or right of $p_k$ (e.g., see Fig.~\ref{fig:contain}). Therefore, $s_j$ covers all points of $P_r(s_i)$.
\qed
\end{proof}

The following lemma will help us to reduce the problem to the 1D
problem.

\begin{lemma}\label{lem:unitopt}
Suppose $S_{opt}$ is an optimal solution subset and $s_i$ is a disk
in $S_{opt}$. Then, the following hold.
\begin{enumerate}
\item
$s_i$ must be a useful disk.
\item
$P_m(s_i)$ has at least one point not covered by any disk of $S_{opt}\setminus \{s_i\}$.
\item
All points of $P_l(s_i)\cup P_r(s_i)$ are covered by the disks of $S_{opt}\setminus \{s_i\}$.
\end{enumerate}
\end{lemma}
\begin{proof}
First of all, since $s_i$ is in $S_{opt}$ and $w_i>0$, $s_i$ must cover a
point $p^*\in P$ that is not covered by any other disk of $S_{opt}$.
Depending on whether $a_l(i)=0$ and whether $a_r(i)=n+1$, there are
several cases.

\begin{itemize}
\item
If $a_l(i)=0$ and $a_r(i)=n+1$, then all points of $P$ are covered by
$s_i$. Therefore, $S_{opt}$ has only one disk, which is $s_i$. Further,
$a_l(i)=0$ and $a_r(i)=n+1$ imply that  $P_l(s_i)= P_r(s_i)=
\emptyset$. Hence, the lemma follows.

\item
If $a_l(i)\neq 0$ and $a_r(i)=n+1$, then some disk $s_j$ of
$S_{opt}\setminus\{s_i\}$ must cover the point $a_l(i)$. Then,
by Lemma~\ref{lem:unitcover}, $s_j$ must cover all points of
$P_l(s_i)$. Hence, $p^*\not\in P_l(s_i)$. Since $a_r(i)=n+1$, we have
$P_r(s_i)=\emptyset$. Thus, $p^*$ is in $P_m(s_i)$. Therefore, the
lemma follows.

\item
If $a_l(i)= 0$ and $a_r(i)\neq n+1$, then the proof is analogous to
the above second case and we omit it.

\item
If $a_l(i)\neq 0$ and $a_r(i)\neq n+1$, then by a similar proof as the
above second case, we know that all points of $P_l(s_i)$ are covered
by a disk of $S_{opt}\setminus\{s_i\}$. Similarly, since $a_r(i)\neq n+1$,
we can show that all points of $P_r(s_i)$ are covered
by a disk of $S_{opt}\setminus\{s_i\}$. This implies that $p^*$ is in
$P_m(s_i)$. Therefore, the lemma follows.
\qed
\end{itemize}
\end{proof}

By Lemma~\ref{lem:unitopt}, to find an optimal solution, it is sufficient to consider only useful disks, and
further, for each useful disk $s_i$, it is sufficient to assume that
it only covers the points of $P_m(s_i)=P[a_l(i)+1,a_r(i)-1]$.
This observation leads to the following approach to reduce our problem to an
instance of the 1D problem.

We assume that the indices $a_l(i)$ and $a_r(i)$ for all $i\in [1,m]$
are known.
For each point $p_i$, we project it vertically on $L$, and let $P'$ be
the set of all projected points. For each useful
disk $s_i$, we create a segment on $L$ whose left endpoint has
$x$-coordinate equal to $x(p_{k+1})$ with $k=a_l(i)$ and whose right
endpoint has $x$-coordinate equal to $x(p_{k'-1})$ with $k'=a_r(i)$,
and the weight of the segment is equal to $w_i$. Let $S'$ be the set
of all segments thus defined.
According to the above discussion, an optimal solution to the 1D
problem on $P'$ and $S'$
corresponds to an optimal solution to our original problem on $P$ and
$S$.
By Theorem~\ref{theo:1d}, the 1D problem can be solved in
$O((n+m)\log(n+m))$ time because $|P'|=n$ and $|S'|\leq m$.

It remains to compute the indices $a_l(i)$ and $a_r(i)$ for all $i\in
[1,m]$, which is done in the following lemma.

\begin{lemma}\label{lem:unita}
Computing $a_l(j)$ and $a_r(j)$ for all $j\in [1,m]$ can be done in
$O((n+m)\log(n+m))$ time.
\end{lemma}
\begin{proof}
We only describe how to compute $a_r(j)$ for all $j\in
[1,m]$, and the algorithm for $a_l(j)$ is similar.

We sweep the plane with a vertical line $l$ from left to right, and
an event happens if $l$ encounters a point of $P$ or a disk center. For
this, we first sort all points of $P$ and all disk centers, in
$O((n+m)\log(n+m))$ time.
During the sweeping, we maintain a list $Q$ of disks $s_i$ whose
centers have been
swept and whose indices $a_r(i)$ have not been computed yet.
$Q$ is just a first-in-first-out queue storing the disks ordered by
their centers from left to right.
Initially, $Q=\emptyset$.

During the sweeping, if $l$ encounters the center of a disk $s_j$, we
add $s_j$ to the rear of $Q$. If $l$ encounters a point $p_i$, then we
process it as follows.  Starting from the front disk $s_j$ of $Q$, we
check whether $s_j$ covers $p_i$.
If yes, then one can verify that every disk
in $Q$ covers $p_i$, and thus in this case we finish
processing $p_i$. Otherwise, we remove $s_j$ from $Q$ and set $a_r(j)=i$,
after which we proceed on
the next disk in $Q$ (if $Q$ becomes $\emptyset$, then we finish
processing $p_i$). If $Q$ is not empty after $p_n$ is processed,
then we set $a_r(j)=n+1$ for all $s_j\in Q$.

The running time of the sweeping algorithm after sorting is $O(n+m)$.
The lemma thus follows. \qed
\end{proof}

With the preceding lemma, we have the following theorem.
\begin{theorem}
The line-constrained disk coverage problem for unit disks is solvable in $O((n+m)\log(n+m))$
time.
\end{theorem}

\section{The $L_1$ case}
\label{sec:l1}

In this case, each disk of $S$ is a diamond, whose boundary is comprised of four edges of slopes $1$ and $-1$, but the diamonds of $S$ may have different radii. We show that the problem
can be solved in $O((n+m)\log(n+m))$ time by similar techniques to the unit-disk case in
Section~\ref{sec:unit}.

For each diamond $s_i\in S$, we still define the two indices $a_l(i)$ and
$a_r(i)$ as well as the three subsets $P_l(s_i)$, $P_r(s_i)$, and
$P_m(s_i)$ in exactly the same way as in Section~\ref{sec:unit}.
We still call $s_i$ a {\em useful} disk if $a_l(i)+1<a_r(i)$.

Although the disks may have different radii, the
geometric properties of the $L_1$ metric guarantee that
Lemma~\ref{lem:unitcover} still applies. The proof is literally the same
as before (indeed, one can verify that the region $D$ must be
contained in the diamond $s_j$; e.g., see Fig.~\ref{fig:containl1} as a
counterpart of Fig.~\ref{fig:contain}), so we omit it.
As Lemma~\ref{lem:unitopt} mainly relies on Lemma~\ref{lem:unitcover},
it also applies here. Consequently, once the indices $a_r(j)$ and
$a_l(j)$ for all $j\in [1,m]$ are known, we can use the same algorithm
as before to find an optimal solution in $O((n+m)\log(n+m))$ time.
The algorithm for computing the indices $a_r(j)$ and $a_l(j)$,
however, is not the same as before in Lemma~\ref{lem:unita}. We
provide a new algorithm in the following lemma.

\begin{figure}[t]
\begin{minipage}[t]{\textwidth}
\begin{center}
\includegraphics[height=1.0in]{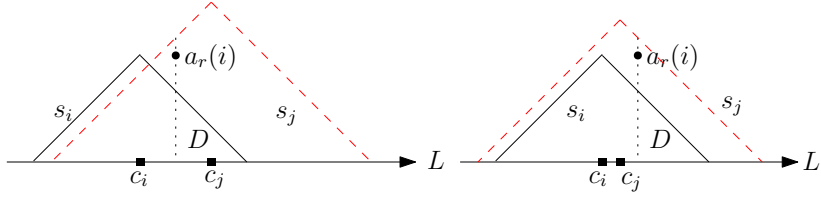}
\caption{\footnotesize Illustrating the proof of Lemma~\ref{lem:unitcover} for the $L_1$ case, as a counterpart of Fig.~\ref{fig:contain}. Now both $s_i$ and $s_j$ are diamonds (only the upper halves are shown). Left: $c_j$ is to the right of point $a_r(i)$. Right: $c_j$ is to the left of point $a_r(i)$. In both cases, $s_j$ contains the region $D$. }
\label{fig:containl1}
\end{center}
\end{minipage}
\vspace{-0.15in}
\end{figure}

\begin{lemma}\label{lem:l1a}
Computing $a_l(j)$ and $a_r(j)$ for all $j\in [1,m]$ can be done in
$O((n+m)\log(n+m))$ time.
\end{lemma}
\begin{proof}
We only describe how to compute $a_r(j)$ for all $i\in
[1,m]$, and the algorithm for $a_l(i)$ is similar.

We sweep the plane with a vertical line $l$ from left to right, and
an event happens if $l$ encounters a point of $P$ or the center of a diamond $s_j$.
For this, we first sort all points of $P$ and the centers of all
diamonds in $O((n+m)\log(n+m))$ time.
During the sweeping, we maintain a list $Q$ of diamonds $s_i$ whose
centers have been swept and whose indices $a_r(i)$ have not been computed yet.
We store the diamonds of $Q$ by a balanced binary search tree with the $x$-coordinates of the
{\em rightmost points} of the diamonds as the keys.
Initially, $Q=\emptyset$.

\begin{figure}[t]
\begin{minipage}[t]{\textwidth}
\begin{center}
\includegraphics[height=1.0in]{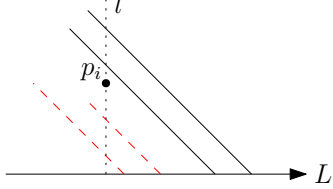}
\caption{\footnotesize Illustrating an event when $l$ encounters a point $p_i\in P$. Four diamonds (only their upper right edges are shown) are in $Q$. To process the event, the two red dashed diamonds will be removed from $Q$, and their indices $a_r(j)$ will be set to $i$.}
\label{fig:sweep1d}
\end{center}
\end{minipage}
\vspace{-0.15in}
\end{figure}

During to the sweeping, if $l$ encounters the center of a diamond
$s_j$, then we insert $s_j$ into $Q$. If $l$ encounters a point $p_i$, then we process it as follows.
Find the diamond $s_j$ in $Q$ with the smallest key (i.e., the diamond of $Q$ whose rightmost point is the leftmost). If $s_j$ covers $p_i$, then one can verify that every diamond
in $Q$ covers $p_i$, and thus in this case we finish
processing $p_i$. Otherwise (e.g., see Fig.~\ref{fig:sweep1d}),
we delete $s_j$ from $Q$ and set $a_r(j)=i$,
after which we proceed on the next diamond in $Q$ with the smallest key
(if $Q$ becomes $\emptyset$, then we finish processing $p_i$).
If $Q$ is not empty after $p_n$ is processed,
then we set $a_r(j)=n+1$ for all $s_j\in Q$.

The running time of the sweeping algorithm after sorting is $O(n+m)$.
The lemma thus follows.
\qed
\end{proof}

\begin{theorem}
The line-constrained disk coverage problem in the $L_1$ metric is solvable in $O((n+m)\log(n+m))$
time.
\end{theorem}

\section{The $L_{\infty}$ and $L_2$ cases}
\label{sec:general}

In this section, we give our algorithms for the $L_{\infty}$ and $L_2$ cases.
The algorithms are similar in the high level. However, the nature of the $L_2$ metric makes the $L_2$ case more involved in the low level computations. In Section~\ref{sec:scheme}, we present a high-level algorithmic scheme that works for both metrics. Then, we complete the algorithms for $L_{\infty}$ and $L_2$ cases in Sections~\ref{sec:linfty} and \ref{sec:l2}, respectively.

\subsection{An algorithmic scheme for $L_{\infty}$ and $L_2$ metrics}
\label{sec:scheme}

In this subsection, unless otherwise stated, all statements are applicable to both metrics.
Note that  a disk in the $L_{\infty}$ metric  is a square.

For a disk $s_k\in S$, we say that a subsequence $P[i,j]$ of $P$ with $1\leq i\leq j\leq n$ is a {\em maximal subsequence covered} by $s_k$ if all points of $P[i,j]$ are covered by $s_k$ but neither $p_{i-1}$ nor $p_{j+1}$ is covered by $s_k$ (it is well defined due to $p_0$ and $p_{n+1}$).
Let $F(s_k)$ be the set of all maximal subsequences covered by $s_k$. Note that the subsequences of $F(s_k)$ are pairwise disjoint.


\begin{lemma}\label{lem:optgeneral}
Suppose $S_{opt}$ is an optimal solution subset and $s_k$ is a disk
of $S_{opt}$. Then, there is a subsequence $P[i,j]$ in $F(s_k)$ such that the following hold.
\begin{enumerate}
\item
$P[i,j]$ has a point that is not covered by any disk in $S_{opt}\setminus\{s_k\}$.
\item
For any point $p\in P$ that is covered by $s_k$ but is not in $P[i,j]$, $p$ is covered by a disk in $S_{opt}\setminus\{s_k\}$.
\end{enumerate}
\end{lemma}
\begin{proof}
First of all, $s_k$ must cover a point $p^*$ that is not covered by any disk in $S_{opt}\setminus\{s_k\}$.
Since the subsequences of $F(s_k)$ are pairwise disjoint, $p^*$ is in
a unique subsequence $P[i,j]$ of $F(s_k)$. In the following, we show
that $P[i,j]$ has the property as stated in the lemma.

Consider any point $p_h\in P$ that is covered by $s_k$ but is not in
$P[i,j]$. By the definition of maximal sequences,
either $h\leq i-1$ or $h\geq j+1$. We only discuss the
case $h\leq i-1$ since the other case is similar. In the following, we
show that $p_h$ must be covered by a disk in
$S_{opt}\setminus\{s_k\}$, which will prove the lemma.

\begin{figure}[h]
\begin{minipage}[t]{\textwidth}
\begin{center}
\includegraphics[height=1.0in]{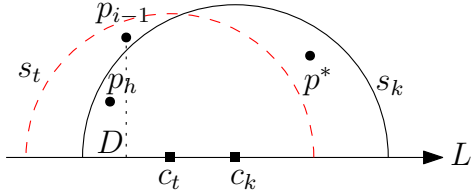}
\caption{\footnotesize Illustrating the proof of Lemma~\ref{lem:optgeneral}. The red dashed half-circle shows disk $s_t$, which covers $p_{i-1}$, and $x(c_t)\leq x(c_k))$. The disk $s_t$ must also cover the point $p_h$.}
\label{fig:cover}
\end{center}
\end{minipage}
\vspace{-0.15in}
\end{figure}

By the definition of maximal sequences, neither $p_{i-1}$ nor $p_{j+1}$ is covered by $s_k$.
Since $S_{opt}$ is an optimal solution,
$S_{opt}\setminus\{s_k\}$ must have a disk $s_t$ that covers
$p_{i-1}$. According to the above discussion, $s_t$ does not cover
$p^*$. Since $p^*$ is to the right of $p_{i-1}$, the center $c_t$ of
$s_t$ cannot be to the right of the center $c_k$ of $s_k$, since
otherwise $s_t$ would cover $p^*$ as well because $s_k$ covers $p^*$.
Let $D$ be the region of $s_k$
to the left of the vertical line through $p_{i-1}$. It is easy to see
that $p_h$ is in $D$ (e.g., see Fig.~\ref{fig:cover}). Since
$x(c_t)\leq x(c_k)$ and $p_{i-1}$ is in $s_t$ but not in $s_k$, one
can verify that $D$ is contained in $s_t$. Thus, $p_h$ must be covered
by $s_t$.
\qed
\end{proof}

In light of Lemma~\ref{lem:optgeneral}, we reduce the problem to an instance of the 1D problem with a point set $P'$ and a line segment set $S'$, as follows.

For each point of $P$, we vertically project it on $L$, and
the set $P'$ is comprised of all such projected points. Thus $P'$ has exactly $n$ points.
For any $1\leq i\leq j\leq n$, we use $P'[i,j]$ to denote the
projections of the points of $P[i,j]$. For each point $p_i\in P$, we use $p_i'$ to denote its projection point in $P'$.

The set $S'$ is defined as follows. For each disk $s_k\in S$ and each
subsequence $P[i,j]\in F(s_k)$, we create a segment for $S'$, denoted
by $s[i,j]$, with left endpoint at $p_i'$ and right endpoint at $p'_j$. Thus,
$s[i,j]$ covers exactly the points of $P'[i,j]$.
We set the weight of $s[i,j]$ to $w_k$. Note that if $s[i,j]$ is already in $S'$, which
is defined by another disk $s_h$, then we only need to
update its weight to $w_k$ in case $w_k<w_h$ (so each segment appears only
once in $S'$). We say that $s[i,j]$ is defined by $s_k$ (resp., $s_h$)
if its weight is equal to $w_k$ (resp., $w_h$).


According to Lemma~\ref{lem:optgeneral}, we intend to say that an
optimal solution $OPT'$ to the 1D problem on $P'$ and $S'$
corresponds to an optimal solution $OPT$ to the original problem on
$P$ and $S$ in the following sense: if a segment $s[i,j]\in S'$ is
included in $OPT'$, then we include the disk that defines $s[i,j]$ in
$OPT$. However, since a disk of $S$ may define multiple
segments of $S'$, to guarantee the correctness of the above correspondence,
we need to show that $OPT'$ is a {\em valid solution}: no
two segments in $OPT'$ are defined by the same disk of $S$. For this,
we have the following lemma.

\begin{lemma}\label{lem:valid}
Any optimal solution on $P'$ and $S'$ is a valid solution.
\end{lemma}
\begin{proof}
Let $OPT'$ be any optimal solution. Let $s[i,j]$ be a segment in $OPT'$.
So $s[i,j]$ is defined by a disk $s_k$ for the maximal subsequence $P[i,j]$.
In the following we show that no other segments defined by $s_k$ are in $OPT'$, which will prove the lemma.

Assume to the contrary that $OPT'$ has another segment $s[i',j']$
defined by $s_k$. Then, since the maximal subsequences covered by $s_k$ are pairwise disjoint,
either $j'<i$ or $j<i'$ holds. In the following, we only discuss the case $j'<i$ since the other case is similar.

By the definition of maximal subsequences, neither $p_{j'+1}$ nor $p_{i-1}$ is covered by $s_k$. Note that $j'+1=i-1$ is possible. Hence, $OPT'$ must have a segment $s'$ defined by
another disk $s_h$ covering $p_{i-1}$ such that $s'$ covers the projection point $p'_{i-1}$ of $p_{i-1}$. Since $s[i,j]$ is in $OPT'$,
$P'[i,j]$ has at least one point $p^*$ that is not covered by any
segment in $OPT'$ other than $s[i,j]$. Thus, $p^*$ is not covered by $s'$.

We claim that the center $c_h$ of $s_h$ is strictly to the left of the center of $c_k$ of $s_k$. Indeed, assume to the contrary that $x(c_h)\geq x(c_k)$. Then, let $D$ be the region of $s_k$ to the right of the vertical line through $p_{i-1}$. Notice that all points of $P[i,j]$ are in $D$. Also, since $s_h$ covers $p_{i-1}$ while $s_k$ does not and $x(c_h)\geq x(c_k)$, $D$ is contained in $s_h$. This means that all points of $P[i,j]$ are covered by $s_h$, and thus all points of $P[i-1,j]$ are covered by $s_h$ since $s_h$ covers $p_{i-1}$. Hence, the segment $s'$ covers all points of $P'[i-1,j]$, and thus, $s'$ covers the points $p^*$, which contradicts with the fact that $s'$ does not cover $p^*$. This proves the claim that $x(c_h)<x(c_k)$.

Depending on whether $s_h$ covers all points of $P[j'+1,i-1]$, there are two cases.

\begin{itemize}
\item
If $s_h$ covers all points of $P[j'+1,i-1]$, then since $x(c_h)<x(c_k)$ and $s_k$ does not cover $p_{j'+1}$ (but covers all points of $P[i',j']$), by the similar analysis as above, we can show that $s_h$ also covers all points of $P[i',j']$ and thus all points of $P[i',i-1]$. This implies that the segment $s'$ covers all projection points of $P'[i',i-1]$. Therefore, if we remove
$s[i',j']$ from $OPT'$, the remaining segments of $OPT'$ still cover all
points of $P'$, which contradicts with that $OPT'$ is an optimal solution.

\item
If $s_h$ does not cover all points of $P[j'+1,i-1]$, then let $h_1$ be the largest index in $[j'+1,i-2]$ such that $p_{h_1}$ is not covered by $s_h$. Then, $p_{h_1}'$ is not covered by the segment $s'$. Hence, $OPT'$ must have a segment defined by another disk $s_{j_1}$ covering $p_{h_1}$ such that the segment covers $p_{h_1}'$. By the same analysis as above, we can show that $x(c_{j_1})< x(c_h)$, and thus $x(c_{j_1})<x(c_k)$.

If $s_{j_1}$ covers all points of $P[j'+1,h_1-1]$, then we can use the same analysis as the above case to show that $s[i',j']$ is a redundant segment of $OPT'$, which incurs contradiction. Otherwise, we let $h_2$ be the largest index in $[j'+1,h_1-1]$ such that $p_{h_2}$ is not covered by $s_{j_1}$. Then, we can follow the same analysis above to either obtain contradiction or consider the next index in $[j'+1,h_2-1]$. Note that this procedure is finite as the number of indices of $[j'+1,h_1-1]$ is finite. Therefore, eventually we will obtain contradiction.
\end{itemize}
The lemma thus follows. \qed
\end{proof}

With the above lemma, combining with our algorithm for the 1D problem,
we have the following result.

\begin{lemma}\label{lem:bothmetrics}
If the set $S'$ is computed, then an
optimal solution can be found in $O((n+|S'|)\log(n+|S'|))$ time.
\end{lemma}

It remains to determine the size of $S'$ and compute $S'$.
An obvious answer is that $|S'|$ is bounded by
$m\cdot \lceil n/2\rceil$ because each disk can have at most $\lceil n/2\rceil$ maximal sequences of
$P$,
and a trivial algorithm can compute $S'$ in $O(nm\log (m+n))$
time by scanning the sorted list $P$ for each disk. Therefore, by
Lemma~\ref{lem:bothmetrics}, we
can solve the problem in both $L_{\infty}$ and $L_2$ metrics  in
$O(nm\log (m+n))$ time.

With more geometric observations, the following two subsections will
prove the two following lemmas, respectively.

\begin{lemma}\label{lem:linftyset}
In the $L_{\infty}$ metric, $|S'|\leq 2(n+m)$ and $S'$ can be computed
in $O((n+m)\log (n+m))$ time.
\end{lemma}

\begin{lemma}\label{lem:l2set}
In the $L_2$ metric, $|S'|\leq 2(n+m)+\kappa$ and $S'$ can be computed
in $O((n+m)\log (n+m) + \kappa\log m)$ time.
\end{lemma}

With Lemma~\ref{lem:bothmetrics}, we have the following results.

\begin{theorem}
The line-constrained disk coverage problem in the $L_{\infty}$ metric
is solvable in $O((n+m)\log(n+m))$ time.
\end{theorem}

\begin{theorem}\label{theo:L2general}
The line-constrained disk coverage problem in the $L_{2}$
metric is solvable in $O(nm\log(m+n))$ time or in $O((n+m)\log(n+m)+\kappa\log m)$ time,
where $\kappa$ is the number of pairs of disks of $S$ that intersect each
other.
\end{theorem}

\paragraph{Bounding couples.} Before moving on, we
introduce a new concept {\em bounding couples}, which will be used to prove Lemmas~\ref{lem:linftyset} and \ref{lem:l2set} in
Sections~\ref{sec:linfty} and \ref{sec:l2}.

Consider a disk $s_k\in S$. Let $p_l(s_k)$ denote the rightmost point of $P\cup \{p_0,p_{n+1}\}$
strictly to the left of $l_k$; similarly,  let
$p_r(s_k)$ denote the leftmost point of $P\cup \{p_0,p_{n+1}\}$
strictly to the right of $r_k$. Let $P(s_k)$ denote the subset of
points of $P$ between $p_l(s_k)$ and $p_r(s_k)$ inclusively that are
outside $s_k$. We sort the points of $P(s_k)$ by their
$x$-coordinates, and we call each adjacent pair of points (or their
indices) in the sorted list a {\em bounding couple} (e.g., see Fig.~\ref{fig:defcouple}).
Let $C(s_k)$ denote the set of all bounding couples of
$s_k$, and for each bounding couple of $C(s_k)$, we assign $w_k$ to it
as the weight. Let $\calC=\bigcup_{1\leq k\leq m}C(s_k)$, and if the
same bounding couple is defined by multiple disks, then we only keep the
copy in $\calC$ with the minimum weight. Also, we consider a bounding couple
$(i,j)$ as an ordered pair such that $i<j$, and $i$ is considered as the left
end of the couple while $j$ is the right end.

\begin{figure}[t]
\begin{minipage}[t]{\textwidth}
\begin{center}
\includegraphics[height=1.0in]{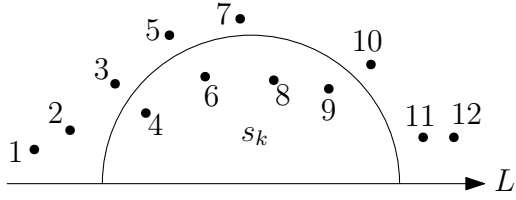}
\caption{\footnotesize Illustrating the definition of bounding couples: the numbers are the indices of the points of $P$. In this example, $p_l(s_k)$ is point $2$ and $p_r(s_k)$ is point $11$, and the bounding couples are: $(2,3)$, $(3,5)$, $(5,7)$, $(7,10)$, $(10, 11)$.}
\label{fig:defcouple}
\end{center}
\end{minipage}
\vspace{-0.15in}
\end{figure}

The reason why we define bounding couples is that if $P[i,j]$ is a maximal
subsequence of $P$ covered by $s_k$ then $(i-1,j+1)$ is a bounding
couple. On the other hand, if $(i,j)$ is a bounding couple of $C(s_k)$,
then $P[i+1,j-1]$ is a maximal subsequence of $P$ covered by $s_k$
unless $j=i+1$. Hence, each bounding couple $(i,j)$ of $\calC$ with $j\neq i+1$
corresponds to a segment in the set $S'$, and $|S'|\leq |\calC|$.
Observe that $\calC$ has at most $n-1$ couples $(i,j)$ with $j=i+1$,
and given $\calC$, we can obtain $S'$ in additional $O(|\calC|)$ time.

According to our above discussion, to prove Lemmas~\ref{lem:linftyset} and \ref{lem:l2set}, it suffices to prove the following two lemmas.

\begin{lemma}\label{lem:linftyc}
In the $L_{\infty}$ metric, $|\calC|\leq 2(n+m)$ and $\calC$ can be computed
in $O((n+m)\log (n+m))$ time.
\end{lemma}

\begin{lemma}\label{lem:l2c}
In the $L_2$ metric, $|\calC|\leq 2(n+m)+\kappa$ and $\calC$ can be computed
in $O((n+m)\log (n+m) + \kappa\log m)$ time.
\end{lemma}


Consider a bounding couple $(i,j)$ of $\calC$, defined by a disk
$s_k$. We call it a {\em left bounding couple} if $p_i=p_l(s_k)$,
a {\em right bounding couple} if $p_j=p_r(s_k)$, and a {\em middle
bounding couple} otherwise (e.g., in Fig.~\ref{fig:defcouple}, $(2,3)$ is the left bounding couple, $(10,11)$ is the right bounding couple, and the rest are middle bounding couples).
It is easy to see that a disk can define at
most one left bounding couple and at most one right bounding couple.
Therefore, the number of left and right bounding couples in $\calC$ is
at most $2m$. It remains to bound the number of middle bounding
couples of $\calC$.

In the following, we will prove Lemmas~\ref{lem:linftyc} and \ref{lem:l2c} in Sections~\ref{sec:linfty} and \ref{sec:l2}, respectively.

\subsection{The $L_{\infty}$ metric}
\label{sec:linfty}

In this section, our goal is to prove Lemma~\ref{lem:linftyc}.

In the $L_{\infty}$ metric, every disk is a square that has four
axis-parallel edges.  We use
$l_k$ and $r_k$ to particularly refer to the left and right endpoints of the upper
edge of $s_k$, respectively.

For a point $p_i$ and a square $s_k$, we say that $p_i$ is {\em
vertically above (resp., below)} the upper edge of $s_k$ if $p_i$ is above (resp., below) the
upper edge of $s_k$ and $x(l_k)\leq x(p_i)\leq x(r_k)$.
Due to our general position assumption, $p_i$ is not on the boundary of $s_k$, and thus
$p_i$ above/below the upper edge of $s_k$ implies that $p_i$ is
strictly above/below the edge. Also, since no point of $P$ is below $L$,
a point $p_i\in P$ is in $s_k$ if and only if $p_i$ is vertically
below the upper edge of $s_k$.
If $p_i$ is vertically above the upper edge of $s_k$, we also say that $p_i$ is vertically above $s_k$ or $s_k$ is vertically below $p_i$.

The following lemma proves an upper bound for $|\calC|$.

\begin{lemma}\label{lem:bound}
$|\calC|\leq 2(n+m)$.
\end{lemma}
\begin{proof}
Recall that the total number of left and right bounding couples of
$\calC$ is at most $2m$.  In the following, we show that the number of middle
bounding couples of $\calC$ is at most $2n$.

We first prove an {\em observation}: For each point $p_j$ of
$P$, among all points of $P$ to the northwest of $p_j$, there is at
most one point that can form a middle bounding couple with $p_j$; similarly,
among all points of $P$ to the northeast of $p_j$, there is at
most one point that can form a middle bounding couple with $p_j$.

We only prove the northwest case since the other
case is analogous. Suppose there is a point $p_i\in P$ to the northwest of $p_j$ and
$(p_i,p_j)$ is a middle bounding couple.
Assume to the contrary that there is another point $p_h\in P$ to the northwest of $p_j$ and
$(p_h,p_j)$ is a middle
bounding couple defined by a disk $s_k$. Without loss of generality, we assume $h<i$.

Since $(p_h,p_j)$ is a middle bounding couple, both $p_h$ and $p_j$
are vertically above $s_k$.
Since $p_i$ is to the northwest of $p_j$ and $h<i<j$,
$p_i$ is also vertically above $s_k$.
But then $p_i$ would prevent $(h,j)$ from being a middle bounding couple
defined by $s_k$, incurring contradiction.
This proves the observation.

We proceed to show that the number
of middle bounding couples is at most $2n$. Indeed, for any middle bounding couple
$(i,j)$ of $\calC$, we charge it to the lower point of $p_i$ and $p_j$. In
light of the observation, each point of $P$ will be charged at most
twice. As such, the total number of middle bounding couples is at most $2n$.
The lemma thus follows. \qed
\end{proof}

We proceed to compute the set $\calC$. The following lemma gives an
algorithm to compute all left and right bounding couples of $\calC$.

\begin{lemma}\label{lem:leftrightcouples}
All left and right bounding couples of $\calC$ can be computed in
$O((n+m)\log(n+m))$ time.
\end{lemma}
\begin{proof}
We only describe how to compute all left bounding couples, and the
algorithm for computing the right bounding couples is similar.

First of all, we compute the points $p_l(s_k)$ and  $p_r(s_k)$ for all $k=1,2,\ldots,
m$. Each such point can be computed in $O(\log n)$ time by binary
search on the sorted sequence of $P$. Hence, computing all such points
takes $O(m\log n)$ time. To compute all left bounding couples, it is
sufficient to compute the points $p(s_k)$ for all disks $s_k\in S$, where $p(s_k)$ is the
leftmost point of $P$ outside $s_k$ and between $l_k$ and $r_k$ if it
exists and $p(s_k)$ is $p_r(s_k)$ otherwise, because
$(p_l(s_k),p(s_k))$ is the left bounding couple defined by $s_k$.
To this end, we propose the following algorithm.

We sweep a vertical line $l$ from left to
right, and an event happens if $l$ encounters a point of
$P\cup\{l_k,r_k|\ 1\leq k\leq m\}$. For this, we first sort all points of
$P\cup\{l_k,r_k|\ 1\leq k\leq m\}$. During the sweeping, we
use a balanced binary search tree $T$ to maintain those disks $s_k$
intersecting $l$ whose points $p(s_k)$ have not been computed yet. The
disks in $T$ are ordered by the $y$-coordinates of their upper edges.

During the sweeping, if $l$ encounters the left endpoint $l_k$ of
a disk $s_k$, we insert $s_k$ into $T$. If $l$
encounters the right endpoint $r_k$ of $s_k$, we remove
$s_k$ from $T$ and set $p(s_k)=p_r(s_k)$. If $l$
encounters a point $p_i$ of $P$, then for each disk $s_k$ of $T$ whose
upper edge is below $p_i$, we set $p(s_k)=p_i$ and remove $s_k$ from $T$.

It is not difficult to see that the algorithm correctly computes all
points $p(s_k)$ for all $s_k\in S$ in $O((n+m)\log (m+n))$ time. The lemma
thus follows. \qed
\end{proof}

In the following, we focus on computing all middle bounding couples of
$\calC$.



\subsubsection{Computing the middle bounding couples}


We sweep a vertical line $l$ from left to right, and an event happens
if $l$ encounters a point in $P\cup\{l_k,r_k|\ 1\leq k\leq m\}$.
Let $H$ be the set of disks that intersect $l$. During the sweeping, we maintain the following information and
invariants (e.g., see Fig.~\ref{fig:invariants}).

\begin{figure}[t]
\begin{minipage}[t]{\textwidth}
\begin{center}
\includegraphics[height=1.5in]{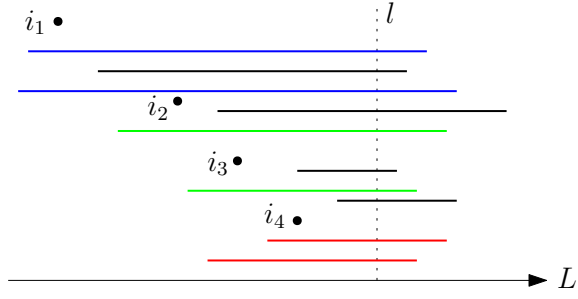}
\caption{\footnotesize Illustrating the information maintained by our sweeping algorithm. $P(l)=\{p_{i_1},p_{i_2},p_{i_3},p_{i_4}\}$. Each horizontal segment represents the upper edge of a disk. $H(i_1)$ consists of two blue disks and $H(i_4)$ consists of two red disks. $H_0$ consists of three black disks.}
\label{fig:invariants}
\end{center}
\end{minipage}
\vspace{-0.15in}
\end{figure}

\begin{enumerate}
\item
A sequence $P(l)=\{p_{i_1},p_{i_2},\ldots,p_{i_t}\}$ of $t$ points of $P$, which are to the left of $l$ and ordered from northwest to southeast.
$P(l)$ is stored in a balanced binary search tree $T(P(l))$.

\item
A collection $\calH$ of $t+1$ subsets of $H$: $H(i_j)$ for $j=0,1,\ldots,t$, which form a partition of $H$, defined as follows.

$H(i_t)$ is the subset of disks of $H$ that are vertically below $p_{i_t}$.
For each $j=t-1, t-2, \ldots, 1$, $H(i_j)$ is the subset of disks of
$H\setminus\bigcup^t_{k=j+1}H(i_k)$ that are vertically below $p_{i_j}$.
$H(i_0)=H\setminus\bigcup^t_{j=1}H(i_j)$. While $H(i_0)$ may be empty, none of
$H(i_j)$ for $1\leq j\leq t$ is empty.

Each set $H(i_j)$ is
maintained by a balanced binary search tree $T(H(i_j))$ ordered by the
$y$-coordinates of the upper edges of the disks.
We have all disks stored in leaves of $T(H(i_j))$, and each
internal node $v$ of the tree also stores a weight equal to the minimum weight
of all disks in the leaves of the subtree rooted at $v$.

\item
For each point $p_{i_j}\in P(l)$, among all points of $P$ strictly between
$p_{i_j}$ and $l$, no point is vertically above any disk of $H(i_j)$.

\item
Among all points of $P$ strictly to the left of $l$, no point is vertically
above any disk of $H(i_0)$.
\end{enumerate}

In summary, our algorithm maintains the following trees: $T(P(l))$, $T(H(i_j))$ for all $j\in [0,t]$.

Initially when $l$ is to the left of all disks and points of $P$, we
have $H=\emptyset$ and $P(l)=\emptyset$.
We next describe how to process events.

If $l$ encounters the left endpoint $l_k$ of a disk $s_k$, we insert
$s_k$ to $H(i_0)$.
The time for processing this event is $O(\log m)$ since $|H(i_0)|\leq m$.

If $l$ encounters the right endpoint $r_k$ of a disk $s_k$, we need to determine which set $H(i_j)$ of $\calH$ contains $s_k$. For this, we associate each right endpoint with its disk in the preprocessing so that it can keep track of which set of $\calH$ contains the disk. Using this mechanism, we can determine the set $H(i_j)$ that contains $s_k$ in constant time.
We then remove $s_k$ from $T(H(i_j))$. If $H(i_j)$ becomes empty and $j\neq 0$, then we remove $p_{i_j}$ from
$P(l)$. One can verify that all algorithm invariants still hold.
The time for processing this event is $O(\log (m+n))$.

\begin{figure}[t]
\begin{minipage}[t]{\textwidth}
\begin{center}
\includegraphics[height=1.5in]{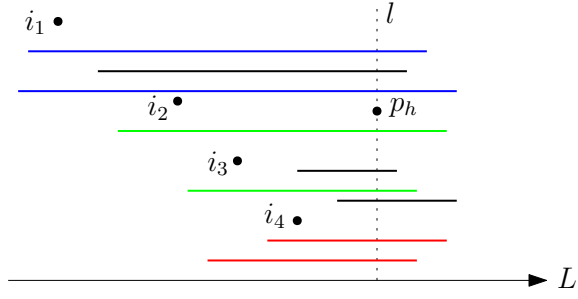}
\caption{\footnotesize Illustrating the processing of an event at $p_h\in P$. In this example, $i_2$, $i_3$, and $i_4$ will be removed from $P(l)$ and $p_h$ will be inserted to $P(l)$, so after the event $P(l)=\{p_{i_1},p_h\}$. Also, $(i_2,h)$, $(i_3,h)$, $(i_4,h)$ will be reported as middle bounding couples.}
\label{fig:algo}
\end{center}
\end{minipage}
\vspace{-0.15in}
\end{figure}

If $l$ encounters a point $p_h$ of $P$, which is a major event we need to handle, we process it as follows.
We search $T(P(l))$ to find the first point $p_{i_j}$ of $P(l)$ below
$p_h$ (e.g., $j=3$ in Fig.~\ref{fig:algo}). We remove the points $p_{i_k}$ for all $k\in [j,t]$ from $P(l)$.
We have the following lemma.

\begin{lemma}\label{lem:pair}
For each point $p_{i_k}$ with $k\in [j,t]$, $(i_k,h)$ is a middle bounding
couple defined by and only by the disks of $H(i_k)$ (i.e., $H(i_k)$ consists of all disks of $S$ that define $(i_k,h)$ as a middle bounding couple).
\end{lemma}
\begin{proof}
By the definition of $H(i_k)$, $p_{i_k}$ is vertically above each disk of $H(i_k)$. By the definition of $j$ and also because all disks of $H(i_k)$ intersect $l$,
$p_h$ is vertically above each disk of $H(i_k)$. With the third
algorithm invariant, $(i_k,h)$ is a middle bounding couple defined by every disk of $H(i_k)$.

On the other hand, suppose a disk $s$ defines $(i_k,h)$ as a middle
bounding couple. Then, both $p_{i_k}$ and
$p_h$ must be vertically above $s$.
This implies that $s$ intersects $l$, and thus $s$ is in $H$.
By algorithm invariant (4), $s$ cannot be in $H(i_0)$.
Because $p_{i_k}$ is vertically above $s$, $s$ must be in $\bigcup_{b=k}^tH(i_b)$.
Further, since $(i_k,h)$ is a middle bounding couple, among all points of $P$ strictly between $p_{i_k}$ and $p_h$, no point is vertically above $s$. This implies that $s$ cannot be in $H(i_b)$ for any $b>k$.
Therefore, $s$ must be in $H(i_k)$. The lemma thus follows. \qed
\end{proof}

In light of Lemma~\ref{lem:pair}, for each $k\in [j,t]$, we report $(i_k,h)$ as a
middle bounding couple with weight equal to the minimum weight of all disks of
$H(i_k)$, which is stored at the root of $T(H(i_k))$.

Next, we process the point $p_{i_{j-1}}$, for which we have the following lemma. The proof technique is similar to that for Lemma~\ref{lem:pair}, so we omit it.

\begin{lemma}
If $p_h$ is vertically below the lowest disk of $H(i_{j-1})$,
then $(i_{j-1},h)$ is not a middle bounding couple; otherwise, $(i_{j-1},h)$ is a middle bounding couple
defined by and only by disks of $H_{j-1}$ that are vertically below $p_h$.
\end{lemma}

By the above lemma, we first check whether $p_h$ is vertically below the lowest disk of
$H(i_{j-1})$. If yes, we do nothing. Otherwise, we
report $(i_{j-1},h)$ as a middle bounding couple with weight equal to
the minimum weight of all disks of $H(i_{j-1})$ vertically below $p_h$, which
can be computed in $O(\log m)$ time by using weights at the internal
nodes of $T(H(i_{j-1}))$. We further have the following lemma.

\begin{lemma}
If all disks of $H(i_{j-1})$ are vertically below $p_h$, then there
does not exist a middle bounding couple $({i_{j-1}},b)$ with $b>h$.
\end{lemma}
\begin{proof}
Assume to the contrary that $(i_{j-1},b)$ is such a middle bounding couple
with $b>h$, say, defined by a disk $s$. Then, since
$x(p_{i_{j-1}})<x(p_h) = x(l) < x(p_b)$,
$s$ intersects $l$, and thus $s$ is in $H$. Also, since $s$ defines
the couple, $p_{i_{j-1}}$ is vertically above $s$.
Note that all disks of $H$ vertically below $p_{i_{j-1}}$
must be in $\bigcup_{k=j-1}^tH(i_k)$, and thus $s$ is in
$\bigcup_{k=j-1}^tH(i_k)$. Recall that all disks of
$\bigcup_{k=j}^tH(i_k)$ are vertically below $p_h$.
Since all disks of $H(i_{j-1})$ are vertically below $p_h$, all disks of
$\bigcup_{k=j-1}^tH(i_k)$ are vertically below $p_h$. Hence, $s$ is also
vertically below $p_h$. Because all three points $p_{i_{j-1}}$, $p_h$, and
$p_b$ are vertically above $s$, and $x(p_{i_{j-1}})<x(p_h)< x(p_b)$,
$(i_{j-1},b)$ cannot be a bounding couple defined by $s$.
The lemma thus follows. \qed
\end{proof}

We check whether $p_h$ is above the highest disk of $H(i_{j-1})$
using the tree $T(H(i_{j-1}))$. If yes, then the above lemma tells that there will be no more middle
bounding couples involving $i_{j-1}$ any more, and thus
we remove $p_{i_{j-1}}$ from $P(l)$.

The following lemma implies that all middle bounding couples with $p_h$
as the right end have been computed.

\begin{lemma}\label{lem:pair1}
For any middle bounding couple $(b,h)$, $b$ must be in $\{i_{j-1},i_j,\ldots,
i_t\}$.
\end{lemma}
\begin{proof}
Assume to the contray that $(b,h)$ is a middle bounding couple with
$b$ not in the set $\{i_{j-1},i_j,\ldots, i_t\}$, say, defined by a
disk $s$. Then, $s$ must intersect $l$, and thus is in $H$. Also, $s$
is vertically below both $p_b$ and $p_h$.

First of all, since $p_b$ is strictly to the left of $l$ and $p_b$ is vertically above $s$,
by our algorithm invariant (4), $s$ cannot be in $H(i_0)$. Thus, $s$ is
in $H(i_j)$ for some $j\in [1,t]$. Depending on whether $i_j<b$, there are two cases.

If $i_j>b$, then since $s\in H(i_j)$, $p_{i_j}$ is vertically above $s$.
Because $x(p_b)<x(p_{i_j})<x(p_h)$ and all these three points are
vertically above $s$, $(b,h)$ cannot be a middle
bounding couple defined by $s$, incurring contradiction.

If $i_j<b$, then since $s\in H(i_j)$ and $p_b$ is vertically above $s$,
we obtain contradiction with our algorithm invariant (3) as $p_b$ is
strictly between $p_{i_j}$ and $l$.
\qed
\end{proof}

Next, we add $p_h$ to the end of the current sequence $P(l)$ (note
that the points $p_{i_k}$ for all $k\in [j,t]$ and possibly
$p_{i_{j-1}}$ have been removed from $P(l)$; e.g., see
Fig.~\ref{fig:algo}). Finally, we need to
compute the tree $T(H(h))$ for the set $H(h)$, which is comprised of
all disks of $H$ vertically below $p_h$ since $p_h$ is the lowest
point of $P(l)$. We compute $T(H(h))$ as follows.

First, starting from an
empty tree, for each $k=t, t-1,\dots, j$ in this order, we merge $T(H(h))$ with the
tree $T(H(i_k))$. Notice that the upper edge of each disk in $T(H(i_{k}))$
is higher than the upper edges of all
disks of $T(H(h))$. Therefore, each such merge operation can be done in
$O(\log m)$ time. Second, for the tree $T(H(i_{j-1}))$, we perform a
split operation to split the disks into those with upper edges above $p_h$ and
those below $p_h$, and then merge those below $p_h$ with
$T(H(h))$ while keeping those above $p_h$ in $T(H(i_{j-1}))$.
The above split and merge operations can be done in
$O(\log m)$ time. Third, we remove those disks below $p_h$ from $H(i_0)$
and insert them to $T(H(h))$. This is done by repeatedly removing the
lowest disk $s$ from $H(i_0)$ and inserting it to $T(H(h))$ until the upper
edge of $s$ is higher than $p_h$.
This completes our construction of the tree
$T(H(h))$.

The above describes our algorithm for processing the event at $p_h$. One can
verify that all algorithm invariants still hold. The running time of this step is
$O((1+k_1+k_2)\log m)$ time, where $k_1$ is the number of points
removed from $P(l)$ (the number of merge operations is at most $k_1$)
and $k_2$ is the number of disks of $H(i_0)$ got removed for
constructing $T(H(h))$. As we sweep the line $l$ from left to right, once
a point is removed from $P(l)$, it will not be inserted again, and
thus the total sum of $k_1$ in the entire algorithm is at most $n$.
Also, once a disk is removed from $H(i_0)$, it will never be inserted
again, and thus the total sum of $k_2$ in the entire algorithm is at
most $m$. Hence, the overall time of the algorithm is
$O((n+m)\log(n+m))$.
This proves Lemma~\ref{lem:linftyc}.

\subsection{The $L_2$ metric}
\label{sec:l2}

In this section, our goal is to prove Lemma~\ref{lem:l2c}.

Recall our general position assumption that no point of $P$ is on
the boundary of a disk of $S$. Also recall that all points of $P$ are
above $L$. In the $L_2$ metric, the two extreme points $l_k$ and $r_k$
of a disk $s_k$ are unique.
For a point $p_i\in P$ and a disk $s_k\in S$, we say that $p_i$ is {\em
vertically above} $s_k$ if $p_i$ is outside $s_k$ and $x(l_k)\leq
x(p_i)\leq x(r_k)$, and $p_i$ is {\em vertically below} $s_k$ if $p_i$ is inside $s_k$.
We also say that $s_k$ is vertically below $p_i$ if $p_i$ is vertically above $s_k$.

The following lemma gives an upper bound for $|\calC|$.

\begin{lemma}\label{lem:boundl2}
$|\calC|\leq 2(n+m)+\kappa$.
\end{lemma}
\begin{proof}
Recall that the left and right bounding couples of $\calC$ is at most
$2m$. Let $\calC_m$ denote the set of all middle
bounding couples of $\calC$. In the following, we argue that $|\calC_m|\leq 2n+\kappa$.

For convenience, we consider a middle bounding couple
$(i,j)$ as a {\em bounding interval} $[i,j]$ defined on indices of $P$. We call the indices
larger than $i$ and smaller than $j$ as the {\em interior} of the
interval.
Those indices smaller than $i$ and larger than $j$ are considered {\em outside} the interval.

We say that two bounding intervals $[a,b]$ and
$[a',b']$ {\em  conflict} if either $a<a'<b<b'$ or
$a'<a<b'<b$. Hence, those two intervals do not conflict if either
they are interior-disjoint or one interval contains the
other. Since two bounding intervals defined by the same disk are
interior-disjoint, they never conflict.

We first prove an observation: {\em For any two disks, there is at most one
pair of conflicting bounding intervals defined by the two disks.}

Assume to the contrary there are two pairs of conflicting bounding
intervals defined by two disks $s$ and $s'$.
Let the first pair be $[a,b]$ and $[a',b']$ and the second
pair be $[c,d]$ and $[c',d']$. Without loss of generality, we assume
that $[a,b]$ and $[c,d]$ are defined by $s$, and $[a',b']$ and
$[c',d']$ are defined by $s'$. Note that $[a,b]$ and $[c,d]$
may be the same and $[a',b']$ and $[c',d']$ may also be the same.
However, as they are different pairs, either $[a,b]$ and $[c,d]$ are
distinct, or $[a',b']$ and $[c',d']$ are distinct. Without loss of
generality, we assume that $[a,b]$ and $[c,d]$ are distinct and $b\leq c$.
Depending on whether $[a',b']$ and $[c',d']$ are the same, there are two
cases.

\begin{itemize}

\item

If $[a',b']$ and $[c',d']$ are the same, then since $b\leq c$, we have
$a<a'<b\leq c<b'<d$ (see Fig.~\ref{fig:conflict}). By the definition of
bounding intervals, $p_b$ and $p_{c}$ are in the disk $s'$ while
$p_{a'}$ and $p_{b'}$ are vertically above $s'$,
and similarly, $p_{a'}$ and $p_{b'}$ are in the disk
$s$ while $p_a,p_b,p_{c},p_{d}$ are vertically above $s$.

\begin{figure}[t]
\begin{minipage}[t]{0.52\textwidth}
\begin{center}
\includegraphics[height=0.5in]{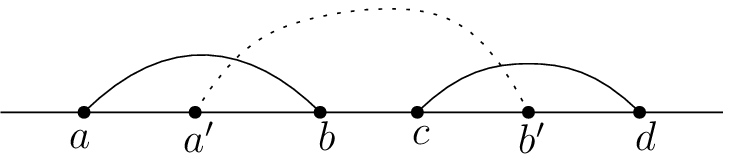}
\caption{\footnotesize Illustrating the conflicting intervals: Each arc represents an interval.}
\label{fig:conflict}
\end{center}
\end{minipage}
\hspace{0.05in}
\begin{minipage}[t]{0.47\textwidth}
\begin{center}
\includegraphics[height=0.7in]{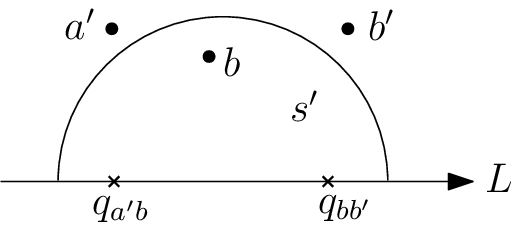}
\caption{\footnotesize Illustrating the disk $s'$ and points $a'$, $b'$, $b$, $q_{a'b}$, and $q_{bb'}$.}
\label{fig:diskcontain}
\end{center}
\end{minipage}
\vspace{-0.15in}
\end{figure}

Since $p_b$ is contained in $s'$ while $p_{a'}$ and $p_{b'}$ are vertically
above $s'$ (e.g., see Fig.~\ref{fig:diskcontain}),
we claim that any disk centered at $L$ and containing both $p_{a'}$ and
$p_{b'}$ must contain the point $p_b$. Indeed, let $q_{a'b}$ be the point
on $L$ that has the same distance with $p_{a'}$ and $p_b$, and
let $q_{bb'}$ be the point on $L$ that has the same distance with $p_b$
and $p_{b'}$ (e.g., see Fig.~\ref{fig:diskcontain}).  Since $x(p_{a'})<x(p_b)$ and $p_b$ is in $s'$ while $p_{a'}$ is not,
we can obtain that $x(q_{a'b})<x(c')$, where $c'$ is the center of $s'$.
For the same reason, $x(q_{bb'})>x(c')$. Therefore, $q_{a'b}$ is strictly to the left of $q_{bb'}$.
Now consider any disk $s''$ with center $c''$ at $L$ such that $s''$ contains both $p_{a'}$ and
$p_{b'}$. If $x(c'')\leq x(q_{a'b})$, then $x(c'')< x(q_{bb'})$ and thus
$c''$ is closer to $p_b$ than to $p_{b'}$. Since $s''$ contains $p_{b'}$, $s''$
also contains $p_b$. On the other hand, if $x(c'')> x(q_{a'b})$, then
$c''$ is closer to $p_b$ than to $p_{a'}$. Since $s''$ contains $p_{a'}$, $s''$
also contains $p_b$. This proves the claim.

Recall that the disk $s$ contains $p_{a'}$ and $p_{b'}$. By the above claim,
$s$ contains $p_b$, but this contradicts with that $p_b$ is strictly
above $s$.

\item
If $[a',b']$ and $[c',d']$ are not the same, then without loss of
generality, we assume that $b'\leq c'$. Since $[a,b]$ conflicts with
$[a',b']$,  either $a<a'<b<b'$ or $a'<a<b'<b$. Similarly, since $[c,d]$ conflicts with
$[c',d']$,  either $c<c'<d<d'$ or $c'<c<d'<d$. In the following,
we assume that $a<a'<b<b'$ and $c<c'<d<d'$ (e.g., see Fig.~\ref{fig:conflict2}), and the other cases can be
proved in a similar way.

\begin{figure}[t]
\begin{minipage}[t]{\textwidth}
\begin{center}
\includegraphics[height=0.5in]{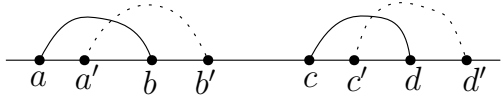}
\caption{\footnotesize Illustrating the conflicting intervals: Each arc represents an interval. The intervals of solid (resp., dotted) arcs are defined by $s$ (resp., $s'$).}
\label{fig:conflict2}
\end{center}
\end{minipage}
\vspace{-0.15in}
\end{figure}

Since $c<c'<d$ and $b'\leq c'$, we obtain that $a'<b<c'$. Since $[a',b']$
and $[c',d']$ are bounding intervals defined by the disk $s'$ while
$b$ is in the interior of $[a',b']$, $s'$
contains $p_b$ but is vertically below $p_{a'}$ and $p_{c'}$. Then, by the
claim proved in the first case, any disk centered at $L$ and
containing both $p_{a'}$ and $p_{c'}$ must contain $p_{b}$ as well.

On the other hand, since $[a,b]$ and $[c,d]$ are bounding intervals
defined by $s$ while $a'$ is in the interior of $[a,b]$ and $c'$ is in
the interior of $[c,d]$, $s$ contains both $p_{a'}$ and $p_{c'}$ but is
vertically below $p_b$. However, since $s$ contains both $p_{a'}$ and
$p_{c'}$ and $s$ is centered at $L$, according to the above claim, $s$
contains $p_b$. Therefore, we obtain contradiction.
\end{itemize}

This proves the observation.

\medskip

We then prove another observation: {\em If a bounding interval defined
by a disk conflicts with a bounding interval defined by another disk,
then the two disks must intersect.}

Indeed, suppose two bounding intervals $[a,b]$ and $[a',b']$ conflict. Let $s$ be the disk defining $[a,b]$ and $s'$ be the disk defining $[a',b']$. Without loss of generality, we assume that $a<a'<b<b'$. By the definition of bounding intervals, $s$ is vertically below $p_a$ and $p_b$, and $s'$ is vertically below $p_c$ and $p_d$. Therefore, both $s'$ and $s$ contain the $x$-interval $[x(p_{a'}),x(p_b)]$ on $L$, and thus they intersect.

\medskip

The above two observations imply that the total number of pairs of conflicting intervals of $\calC_m$ is at most $\kappa$. Now, for each pair of conflicting intervals, we remove one interval from $\calC_m$, so we remove at most $\kappa$ intervals from $\calC_m$. For differentiation, let $\calC_m'$ denote the new set of $\calC_m$ after the removal, and $\calC_m$ still refers to the original set. Observe that $|\calC_m|\leq |\calC'_m|+\kappa$ and no two intervals of $\calC'_m$ conflict. In the
following we show $|\calC'_m|\leq 2n$, which will lead to $|\calC_m|\leq \kappa+2n$.

Our proof mainly relies on the property that no two bounding intervals of $\calC'_m$ conflict. For any two intervals of $\calC'_m$, either they are interior-disjoint or one contains the other. We will form all intervals of $\calC'_m$ as a tree structure $T$. To this end,
for each $i$ with $1\leq i\leq n-1$, if $[i,i+1]$ is not in $\calC'_m$, then we add it to $\calC'_m$.
The tree $T$ is defined as follows. Each interval of $\calC'_m$ defines a node of $T$. The $n-1$ intervals $[i,i+1]$ for all $i=1,2, \ldots, n-1$ are the leaves of $T$. For every two intervals $I_1$ and $I_2$ of $\calC'_m$, $I_1$ is the parent of $I_2$ if and only if $I_1$ contains $I_2$ and there is no other interval $I$ in $\calC_m$ such that $I_2\subseteq I\subseteq I_1$.
Notice that every internal node of $T$ has at least two children. Since $T$ has $n-1$ leaves, the number of internal nodes is no more than $n-2$. Therefore, $T$ has no more than $2n$ nodes, implying that $|\calC'_m|\leq 2n$.
\qed
\end{proof}

We next describe our algorithm for computing the set $\calC$.
For each disk $s_k$, we refer to the half-circle of the boundary of $s_k$ above $L$ as the {\em arc} of $s_k$. Note that every two arcs of $S$ intersect at most once.
In the following, depending on the context, $s_k$ may also refer to its arc.

We begin with computing the left and right bounding couples.

\begin{lemma}\label{lem:leftrightpairs}
All left and right bounding couples of $\calC$ can be computed in
$O((n+m)\log(n+m)+\kappa\log m)$ time.
\end{lemma}
\begin{proof}
We only describe how to compute all left bounding couples, because the
algorithm for computing the right bounding couples is similar.

First of all, we compute the points $p_l(s_k)$ and $p_r(s_k)$ for all $1\leq k\leq
m$. Each such point can be computed in $O(\log n)$ time by binary
search on the sorted sequence of $P$. Hence, computing all such points
takes $O(m\log n)$ time. To compute all left bounding couples, it is
sufficient to compute the points $p(s_k)$ for all disks $s_k\in S$, where $p(s_k)$ is the
leftmost point of $P$ outside $s_k$ and between $l_k$ and $r_k$ if it
exists, and $p(s_k)$ is $p_r(s_k)$ otherwise, because
$(p_l(s_k),p(s_k))$ is the left bounding couple defined by $s_k$.
To this end, we propose a sweeping algorithm similar to that for the $L_{\infty}$ case. The difference is that the arcs of $S$ may intersect each other and thus the sweeping needs to handle the events at intersections.

We sweep a vertical line $l$ from left to
right, and an event happens if $l$ encounters a point of
$P\cup\{l_k,r_k|\ 1\leq k\leq m\}$ or an intersection of two arcs of $S$.
For this, we first sort all points of $P\cup\{l_k,r_k|\ 1\leq k\leq m\}$. We determine the intersections and handle the intersection events in a similar way as the sweeping algorithm for computing line segment intersections~\cite{ref:BentleyAl79,ref:BrownCo81,ref:deBergCo08}; note that we are able to do so because every two arcs of $S$ intersect at most once. During the sweeping, we maintain
the arcs $s_k$ of $S$ intersecting $l$ whose points
$p(s_k)$ have not been computed yet. Those arcs are stored in a balanced binary search tree $T$, ordered by the $y$-coordinates of their intersections with $l$.

During the sweeping, if $l$ encounters the left endpoint $l_k$ of
an arc $s_k$, then we insert $s_k$ into $T$. If $l$
encounters the right endpoint $r_k$ of an arc $s_k$, then we remove
$s_k$ from $T$ and set $p(s_k)=p_r(s_k)$. If $l$
encounters a point $p_i$ of $P$, then for each arc $s_k$ of $T$ that is
below $p_i$, we set $p(s_k)=p_i$ and remove $s_k$ from $T$.
If $l$ encounters an intersection of two arcs, then we process it in the same way as the line segment intersection algorithm, and we omit the discussion here (we also need to detect intersections in other events above, which is similar to the line segment intersection algorithm and is omitted)

The running time of the algorithm is $O((n+m)\log(n+m)+\kappa\log m)$. In particular, the $O(\kappa\log m)$ factor in the time complexity is for handling the intersections of the arcs.
\qed
\end{proof}

It remains to compute the middle bounding pairs of $\calC$. The algorithm is similar in spirit to that for the $L_{\infty}$ case. However, it is more involved and requires new techniques due to the nature of the $L_2$ metric as well as the intersections of the disks of $S$.

We sweep a vertical line $l$ from left to right, and an event happens if $l$ encounters a point in $P\cup\{l_k,r_k|\ 1\leq k\leq m\}$ or an intersection of two disk arcs. Let $H$ be the set of arcs that intersect $l$. During the sweeping, we maintain the following information and invariants (e.g., see Fig.~\ref{fig:invariants2}).

\begin{figure}[t]
\begin{minipage}[t]{\textwidth}
\begin{center}
\includegraphics[height=1.5in]{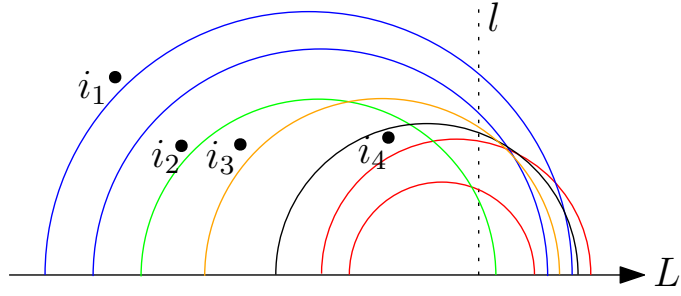}
\caption{\footnotesize Illustrating the information maintained by our sweeping algorithm. $P(l)=\{p_{i_1},p_{i_2},p_{i_3},p_{i_4}\}$. $H(i_1)$ consists of the two blue arcs and $H(i_4)$ consists of the two red arcs. $H(i_0)$ consists of the black arc.}
\label{fig:invariants2}
\end{center}
\end{minipage}
\vspace{-0.15in}
\end{figure}

\begin{enumerate}
\item
A sequence $P(l)=\{p_{i_1},p_{i_2},\ldots,p_{i_t}\}$ of $t$ points to the left of $l$ that are sorted from
left to right.  $P(l)$ is maintained by a balanced binary search tree $T(P(l))$.

\item
A collection $\calH$ of $t+1$ subsets of $H$: $H(i_j)$ for $j=0,1,\ldots,t$, which form a partition of $H$, defined as follows.

$H(i_t)$ is the set of disks of $H$ vertically below $p_{i_t}$. For
each $j=t-1, t-2, \ldots, 1$, $H(i_j)$ is the set of disks of
$H\setminus\bigcup^t_{k=j+1}H(i_k)$ vertically below $p_{i_j}$.
$H(i_0)=H\setminus\bigcup^t_{j=1}H(i_j)$. While $H(i_0)$ may be empty, none of
$H(i_j)$ for $1\leq j\leq t$ is empty.

Each set $H(i_j)$ for $j\in [0,t]$ is maintained by a balanced binary search tree $T(H(i_j))$ ordered by the
$y$-coordinates of the intersections of $l$ with the arcs of the disks.
We have all disks stored in the leaves of the tree, and each internal node $v$ of the tree stores a weight that is equal to the minimum weight of all disks in the leaves of the subtree rooted at $v$.


For each subset $H'\subseteq H$, the arc of $H'$ whose intersection with $l$ is the lowest is called {\em the lowest arc} of $H'$. We maintain a set $H^*$ consisting of the lowest arcs of all sets $H(i_k)$ for $1\leq k\leq t$. So $|H^*|=t$. We use a binary search tree $T(H^*)$ to store disks of $H^*$, ordered by the $y$-coordinates of their intersections with $l$.

\item
For each point $p_{i_j}\in P(l)$, among all points of $P$ strictly between
$p_{i_j}$ and $l$, no point is vertically above any disk of
$H(i_j)$.

\item
Among all points of $P$ strictly to the left of $l$, no point is vertically
above any disk of $H(i_0)$.
\end{enumerate}

\paragraph{Remark.} Our algorithm invariants are essentially the same as those in the $L_{\infty}$ case. One difference is that the points of $P(l)$ are not sorted simultaneously by $y$-coordinates, which is due to that the arcs of $S$ may cross each other (in contrast, in the $L_{\infty}$ case the upper edges of the squares are parallel). For the same reason, for two sets $H(i_k)$ and $H(i_j)$ with $1\leq k< j\leq t$, it may not be the case that all arcs of $H(i_k)$ are above all arcs of $H(i_j)$ at $l$. Therefore, we need an additional set $H^*$ to guide our algorithm, as will be clear later.
\medskip

In our sweeping algorithm, we use similar techniques as the line segment intersection algorithm~\cite{ref:BentleyAl79,ref:BrownCo81,ref:deBergCo08} to determine and handle arc intersections of $S$ (we are able to do so because every two arcs of $S$ intersect at most once), and the time on handling them is $O((m+\kappa)\log m)$.
Below we will not explicitly explain how to handle arc intersections.
Initially $H=\emptyset$ and $l$ is to the left of all arcs of $S$ and all points of $P$.


If $l$ encounters the left endpoint of an arc $s_k$, we insert $s_k$ to $H(i_0)$.

If $l$ encounters the right endpoint $r_k$ of an arc $s_k$, then we need to determine which set of $\calH$ contains $s_k$. For this, as in the $L_{\infty}$ case, we associate each right endpoint with the arc. Using this mechanism, we can find the set $H(i_j)$ of $\calH$ that contains $s_k$ in constant time. Then, we remove $s_k$ from $H(i_j)$. If $j=0$, we are done for this event. Otherwise, if $s_k$ was the lowest arc of $H(i_j)$ before the above remove operation, then $s_k$ is also in $H^*$ and we remove it from $H^*$. If the new set $H(i_j)$ becomes empty, then we remove $p_{i_j}$ from $P(l)$. Otherwise, we find the new lowest arc from $H(i_j)$ and insert it to $H^*$. Processing this event takes $O(\log (n+m))$ time using the trees $T(H^*)$, $T(P(l))$, and $T(H(i_j))$.

If $l$ encounters an intersection $q$ of two arcs $s_a$ and $s_b$, in addition to the processing work for computing the arc intersections, we do the following. Using the right endpoints, we find the two sets of $\calH$ that contain $s_a$ and $s_b$, respectively. If $s_a$ and $s_b$ are from the same set $H(i_j)\in \calH$, then we switch their order in the tree $T(H({i_j}))$. Otherwise, if $s_a$ is the lowest arc in its set and $s_b$ is also the lowest arc in its set, then both $s_a$ and $s_b$ are in $H^*$, so we switch their order in $T(H^*)$. The time for processing this event is $O(\log m)$.

If $l$ encounters a point $p_h$ of $P$, which is a major event we need to handle, we process it as follows. As in the $L_{\infty}$ case, our goal is to determine the middle bounding couples $(i,h)$ with $p_i\in P(l)$.

Using $T(H^*)$, we find the lowest arc $s_k$ of $H^*$.
Let $H(i_j)$ for some $j\in [1,t]$ be the set that contains $s_k$, i.e., $s_k$ is the lowest arc of $H(i_j)$. If $p_h$ is above $s_k$, then we can show that $({i_j},h)$ is a middle bounding couple defined by and only by the arcs of $H(i_j)$ below $p_h$ (e.g., see Fig.~\ref{fig:algo2}). The proof is similar to Lemma~\ref{lem:pair}, so we omit the details. Hence, we report $({i_j},h)$ as a middle bounding couple with weight equal to the minimum weight of all arcs of $H(i_j)$ below $p_h$, which can be found in $O(\log m)$ time using $T(H({i_j}))$. Then, we split $T(H({i_j}))$ into two trees by $p_h$ such that the arcs above $p_h$ are still in $T(H({i_j}))$ and those below $p_h$ are stored in another tree (we will discuss later how to use this tree).
Next we remove $s_k$ from $H^*$. If the new set $H(i_j)$ after the split operation is not empty, then we find its lowest arc and insert it into $H^*$; otherwise, we remove $p_{i_j}$ from $P(l)$. We then continue the same algorithm on the next lowest arc of $H^*$.

\begin{figure}[t]
\begin{minipage}[t]{\textwidth}
\begin{center}
\includegraphics[height=1.5in]{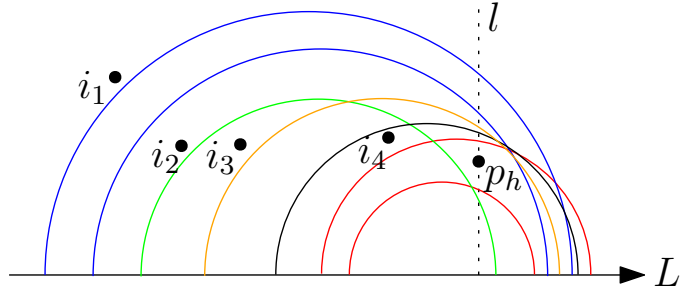}
\caption{\footnotesize Illustrating the processing of an event at $p_h\in P$: $(i_2,h)$ and $(i_4,h)$ will be reported as middle bounding couples, point $i_2$ will be removed from $P(l)$ and $p_h$ will be inserted to $P(l)$.}
\label{fig:algo2}
\end{center}
\end{minipage}
\vspace{-0.15in}
\end{figure}

The above discusses the case where $p_h$ is above $s_k$. If $p_h$ is not above $s_k$, then we are done with processing the arcs of $H^*$. We can show that all middle bounding couples $(b,h)$ with $h$ as the right end have been computed. The proof is similar to Lemma~\ref{lem:pair1}, and we omit the details.

Finally, we add $p_h$ to the rear of $P(l)$. As in the $L_{\infty}$ case, we need to compute the tree $T(H(h))$ for the set $H(h)$, which is comprised of all arcs of $H$ below $p_h$, as follows.

Initially we have an empty tree $T(H(h))$. Let $H'$ be the subset of the arcs of $H^*$ vertically below $p_h$; here $H^*$ refers to the original set at the beginning of the event for $p_h$. The set $H'$ has already been computed above. Let $\calH'$ be the subcollection of $\calH$ whose lowest arcs are in $H'$. We process the subsets $H(i_j)$ of $\calH'$ in the inverse order of their indices (for this, after identifying $\calH'$, we can sort the subsets $H(i_j)$ of $\calH'$ by their indices in $O(|H'|\log m)$ time; note that $|H'|=|\calH'|$), i.e., the subset of $\calH'$ with the largest index is processed first.

Suppose we are processing a subset $H(i_j)$ of $\calH'$. Let $s$ be the lowest arc of $H(i_j)$. 
Recall that we have performed a split operation on the tree $T(H(i_j))$ to obtain another tree consisting of all arcs of $H(i_j)$ below $p_h$, and we use $H'(i_j)$ to denote the set of those arcs and use $T(H'(i_j))$ to denote the tree. If $T(H(h))$ is empty, then we simply set $T(H(h))=T(H'(i_j))$. Otherwise, we find the highest arc $s'$ of $T(H(h))$ at $l$. If $s$ is above $s'$ at $l$, then every arc of $T(H'(i_j))$ is above all arcs of $T(H(h))$ at $l$ and thus we simply perform a merge operation to merge $T(H'(i_j))$ with $T(H(h))$ (and we use $T(H(h))$ to refer to the new merged tree). Otherwise, we call $(s,s')$ an {\em order-violation pair}. In this case, we do the following. We remove $s$ from $T(H'(i_j))$ and insert it to $T(H(h))$. If $T(H'(i_j))$ becomes empty, then we finish processing $H(i_j)$. Otherwise, we find the new lowest arc of $T(H'(i_j))$, still denoted by $s$, and then process $s$ in the same way as above.

The above describes our algorithm for processing a subset $H(i_j)$ of $\calH'$. Once all subsets of $\calH'$ are processed, the tree $T(H(h))$ for the set $H(h)$ is obtained.

After processing the arcs of $H^*$ as above, we also need to consider the arcs of $H(i_0)$. For this, we simply scan the arcs from low to high using the tree $T(H(i_0))$, and for each arc $s$, if $s$ is above $p_h$, then we stop the procedure; otherwise, we remove $s$ from $T(H(i_0))$ and insert it to $T(H(h))$.

This finishes our algorithm for processing the event at $p_h$. The runtime of this step is
$O((1+k_1+k_2+k_3)\cdot \log m)$ time, where $k_1$ is the number of middle bounding couples reported
(the number of merge and split operations is at most $k_1$; also, $|H'|=k_1$),
$k_2$ is the number of arcs of $H(i_0)$ got removed for
constructing $T(H(h))$, and $k_3$ is the number of order-violation pairs.
By Lemma~\ref{lem:boundl2}, the total sum of $k_1$ is at most $2(n+m)+\kappa$ in the entire algorithm.
As in the $L_{\infty}$ case, the total sum of $k_2$ is at most $m$ in the entire algorithm.
The following lemma proves that the total sum of $k_3$ is at most $\kappa$.
Therefore, the overall time of the algorithm is
$O((n+m)\log(n+m)+\kappa\log m)$.

\begin{lemma}\label{lem:exchangebound}
The total number of order-violation pairs in the entire algorithm is at
most $\kappa$.
\end{lemma}
\begin{proof}
We follow the notation defined above. Consider an order-violation pair $(s,s')$, which appears when we process
a subset $H'(i_j)$ of $\calH'$ for constructing $T(H(h))$ during an event at a point $p_h\in P$, such that
$s\in H'(i_j)$ and $s'\in T(H(h))$.
Without loss of generality, we assume that this is the first time that $(s,s')$\footnote{We consider $(s,s')$ as an unordered pair, so $(s,s')$ is the same as $(s',s)$.} appears as an order-violation pair in our entire algorithm. As we process the subsets of $\calH'$ by their inverse index order, $s'$ is from $H(i_k)$ for some $k$ with $j<k\leq t$. Since $(s,s')$ is an order-violation pair, by definition, $s'$ is strictly above $s$ at $x(l)=x(p_h)$; e.g., see Fig.~\ref{fig:ordervio}. On the
other hand, since $s'\in H(i_k)$, we know that $p_{i_k}$ is vertically above $s'$.
Since $s\in H(i_j)$ with $j<k$, $p_{i_k}$ must be vertically below $s$.
Thus, $s$ is strictly above $s'$ at $x(p_{i_k})$. This implies that $s$ and $s'$ has an
intersection strictly between $p_{i_k}$ and $p_h$. We charge the
 pair $(s,s')$ to that intersection.
Because $s$ and $s'$ can have only one intersection,
in the following we show that $(s,s')$ will never appear as an order-violation pair again in the future algorithm.

\begin{figure}[t]
\begin{minipage}[t]{\textwidth}
\begin{center}
\includegraphics[height=1.0in]{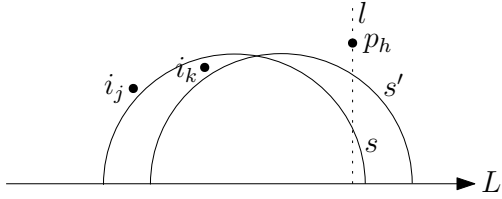}
\caption{\footnotesize Illustrating the proof of Lemma~\ref{lem:exchangebound}: the point $i_k$ is vertically below $s$ but vertically above $s'$.}
\label{fig:ordervio}
\end{center}
\end{minipage}
\vspace{-0.15in}
\end{figure}

First of all, according to our algorithm, $(s,s')$ will not appear as an order-violation pair again during processing the event at $p_h$. After the event, both
$s$ and $s'$ are in $H(h)$. Consider a future event for
processing another point $p_{h'}\in P$. By our algorithm invariant (2), we have
a collection $\calH$ of sets $H_{i'_j}$ with $j=0, 1, \ldots, t'$. Assume to
the contrary that $(s,s')$ appears as an order-violation pair again. Then, $s$ and $s'$ must be from two different sets of
$\calH$, e.g., $H_{i'_j}$ and $H_{i'_k}$.
Without loss of generality, let
$j<k$. By the same analysis as before, we can obtain that $s$ and $s'$
have an intersection $q$ strictly between $p_{i'_j}$ and $p_{h'}$.
Since both $s$ and $s'$ were in $H(h)$ right after the event at $p_h$, it must hold that
$x(p_h)\leq x(p_{i'_j})$. Hence, $x(p_h)<x(q)$.
But this incurs contradiction because we have shown before that the only intersection between $s$ and $s'$ is strictly to the left of $p_h$.

The above shows that  $(s,s')$ will appear as an order-violation pair exactly once in the entire algorithm,
which is charged to their only intersection.
Therefore, the total number of order-violation pairs in the entire algorithm is at most $\kappa$.
\qed
\end{proof}

In summary, all middle bounding couples of $\calC$ can be computed in
$O((n+m)\log(n+m)+\kappa\log m)$ time. Combining with
Lemmas~\ref{lem:boundl2} and \ref{lem:leftrightpairs}, Lemma~\ref{lem:l2c} is proved.

\section{The line-separable unit-disk coverage and the half-plane coverage}
\label{sec:line-separable}

In this section, we show that our techniques for the line-constrained disk coverage problems can also be used to solve other geometric coverage problems.

Recall that the line-separable unit-disk coverage problem refers to the case in which $P$ and centers of $S$ are separated by a line $\ell$ and all disks of $S$ have the same radius. Without loss of generality, we assume that $\ell$ is the $x$-axis and all points of $P$ are above $\ell$. Hence, for each disk $s_i$ of $S$, the portion of $s_i$ above $\ell$ is a subset of its upper half disk. Since disks of $S$ have the same radius, the boundaries of any two disks intersect at most once above $\ell$. We define $\kappa$ as the number of pairs of disks that intersect above $\ell$. Due to the above properties, to solve the problem, we can simply use the same algorithm in Section~\ref{sec:general} for the line-constrained $L_2$ case. Indeed, one can verify that the following critical lemmas that the algorithm relies on still hold: Lemmas~\ref{lem:optgeneral}, \ref{lem:valid}, \ref{lem:boundl2}, \ref{lem:leftrightpairs}, and \ref{lem:exchangebound}. By Theorem~\ref{theo:L2general}, we obtain the following.

\begin{theorem}\label{theo:line-sep}
Given in the plane a set $P$ of $n$ points and a set $S$ of $m$ weighted unit-disks such that $P$ and centers of disks $S$ are separated by a line $\ell$, one can compute a minimum weight disk coverage for $P$ in $O(nm\log(m+n))$ time or in $O((n+m)\log(n+m)+\kappa\log m)$ time,
where $\kappa$ is the number of pairs of disks of $S$ that intersect in the side of $\ell$ containing $P$.
\end{theorem}

\paragraph{\bf Remark.}
Note that although disks of $S$ have the same radius, because their centers may not be on the same line, one can verify that Lemma~\ref{lem:unitcover} does not hold any more. Hence, we can not use the same algorithm as in Section~\ref{sec:unit} for the line-constrained unit-disk case. But if the centers of all disks of $S$ lie on the same line parallel to $\ell$ (and below $\ell$), then Lemma~\ref{lem:unitcover} will hold and thus we can use  the same algorithm as in Section~\ref{sec:unit} to solve the problem in $O((n+m)\log (n+m))$ time.
\bigskip

We now consider the half-plane coverage problem. Given in the plane a set $P$ of $n$ points and a set $S$ of weighted half-planes, the goal is compute a minimum weight half-plane coverage for $P$, i.e., compute a subset of half-planes to cover all points of $P$ so that the total sum of the weights of the half-planes in the subset is minimized.

We start with the {\em lower-only case} where all half-planes of $S$ are lower ones. The problem can be reduced to the line-separable unit-disk coverage problem. Indeed, we first find a horizontal line $\ell$ below all points of $P$. Then, since each half-plane $h$ of $S$ is a lower one, $h$ can be considered as a disk of infinite radius with center below $\ell$. In this way, $S$ becomes a set of unit-disks whose centers are below $\ell$. By Theorem~\ref{theo:line-sep}, we have the following result.

\begin{theorem}\label{theo:loweronly}
Given in the plane a set $P$ of $n$ points and a set $S$ of $m$ weighted lower half-planes, one can compute a minimum weight half-plane coverage for $P$ in $O(nm\log(m+n))$ time or in $O(n\log n+ m^2\log m)$ time.
\end{theorem}

For the general case where $S$ may contain both lower and upper half-planes, we reduce it to a set of $O(n^2)$ instances of the lower-only case, as follows.

Let $S_{opt}$ denote the subset of $S$ in an optimal solution. Har-Peled and Lee~\cite{ref:Har-PeledWe12} observed that if the half-planes of $S_{opt}$ together cover the entire plane then the size of $S_{opt}$ is $3$; in this case, we can enumerate all triples of $S$ and thus obtain an optimal solution in $O(n^3)$ time.

In the following we consider the case where the union of the half-planes of $S_{opt}$ does not cover the entire plane. In this case, the complement of the union of the half-planes of $S_{opt}$ is a (possibly unbounded) convex polygon $R$~\cite{ref:Har-PeledWe12}. For the ease of discussion, we assume that $R$ is bounded since the algorithm for the other case is similar. Let $a$ and $b$ refer to the leftmost and rightmost vertices of $R$, respectively. Let $P_1$ denote the subset of points of $P$ below the line through $a$ and $b$, and $P_2=P\setminus P_1$. The two vertices $a$ and $b$ together partition the edges of $R$ into two chains, a lower chain and an upper chain. Observe that the half-planes that are bounded by the supporting lines of the edges in the lower chain are all lower half-planes and they together cover $P_1$; similarly, the half-planes that are bounded by the supporting lines of the edges of the upper chain are all upper half-planes and they together cover $P_2$. In light of the observation, finding a minimum weight coverage for $P$ is equivalent to solving the following two lower-only case sub-problems: finding a minimum weight coverage for $P_1$ using lower half-planes of $S$ and finding a minimum weight coverage for $P_2$ using upper half-planes of $S$. Because we do not know $P_1$ and $P_2$, we enumerate all possible partitions of $P$ by a line. Clearly, there are $O(n^2)$ such partitions. Hence, solving the half-plane coverage problem for $P$ and $S$ is reduced to $O(n^2)$ instances of the lower-only case. By Theorem~\ref{theo:loweronly}, we can obtain the following result.

\begin{theorem}\label{theo:halfplane}
Given in the plane a set $P$ of $n$ points and a set $S$ of $m$ weighted half-planes, one can compute a minimum weight half-plane coverage for $P$ in $O(n^3m\log(m+n))$ time or in $O(n^3\log n+ n^2m^2\log m)$ time.
\end{theorem}

\section{Concluding remarks}
\label{sec:conclude}

We show that our line-constrained disk coverage problem has an $\Omega((m+n)\log (m+n))$ time lower bound in the algebraic decision tree model even for the 1D case. To this end, in the following we prove that $\Omega(N\log N)$ is a lower bound with $N=\max\{m,n\}$, which implies the $\Omega((m+n)\log (m+n))$ lower bound as $N=\Theta(n+m)$.

The reduction is from the element uniqueness problem. Let $X=\{x_1,x_2,\ldots,x_n\}$ be a set of $n$ numbers, as an instance of the element uniqueness problem. We create an instance of the 1D disk coverage problem with a point set $P$ and a segment set $S$ on the $x$-axis $L$ as follows. For each $x_i\in X$, we create a point on $L$ with $x$-coordinate equal to $x_i$ and create a segment on $L$ which is the above point with weight equal to $1$. Let $P$ be the set of all such points and let $S$ be the set of all such segments. Then, $|P|=|S|=n$, and thus $N=n$. It is not difficult to see that the numbers of $X$ are distinct if and only if the optimal objective value of the 1D disk coverage problem is equal to $n$. As the element uniqueness problem has an $\Omega(n\log n)$ time lower bound under the algebraic decision tree model, our 1D disk coverage problem has an $\Omega(N\log N)$ time lower bound.

The lower bound implies that our algorithms for the 1D, unit-disk, $L_1$, and $L_{\infty}$ cases are all optimal. However, it remains open whether faster algorithms exist for the $L_2$ case. Another direction is to investigate whether the $L_2$ case is 3SUM-hard; if yes, then it is quite likely that our algorithm is nearly optimal.


%

\bibliographystyle{plain}
\bibliography{reference}





\end{document}